\pdfoutput=1
\RequirePackage{ifpdf}
\ifpdf 
\documentclass[pdftex]{sigma}
\else
\documentclass{sigma}
\fi

\numberwithin{equation}{section}

\newtheorem{Theorem}{Theorem}[section]
\newtheorem{Lemma}[Theorem]{Lemma}
\newtheorem{Proposition}[Theorem]{Proposition}
 { \theoremstyle{definition}
\newtheorem{Definition}[Theorem]{Definition}
\newtheorem{Example}[Theorem]{Example}
\newtheorem{Remark}[Theorem]{Remark} }

\newcommand{\ialpha}{i_{\alpha}}
\newcommand{\jbeta}{j_{\beta}}
\newcommand{\kgamma}{k_{\gamma}}
\newcommand{\jalpha}{j_{\alpha}}
\newcommand{\kbeta}{k_{\beta}}
\newcommand{\pgamma}{p_{\gamma}}
\newcommand{\malpha}{m_{\alpha}}
\newcommand{\mbeta}{m_{\beta}}
\newcommand{\mgamma}{m_{\gamma}}
\newcommand{\lbeta}{l_{\beta}}
\newcommand{\lalpha}{l_{\alpha}}
\newcommand{\palpha}{p_{\alpha}}
\newcommand{\kalpha}{k_{\alpha}}
\newcommand{\pbeta}{p_{\beta}}
\newcommand{\ngamma}{n_{\gamma}}
\newcommand{\igamma}{i_{\gamma}}
\newcommand{\jgamma}{j_{\gamma}}
\newcommand{\lgamma}{l_{\gamma}}
\newcommand{\nalpha}{n_{\alpha}}
\newcommand{\hgamma}{h_{\gamma}}
\newcommand{\halpha}{h_{\alpha}}
\newcommand{\idelta}{i_{\delta}}
\newcommand{\jdelta}{j_{\delta}}
\newcommand{\kepsilon}{k_{\epsilon}}
\newcommand{\rhoepsilon}{\rho_{\epsilon}}
\newcommand{\rhogamma}{\rho_{\gamma}}
\newcommand{\rhodelta}{\rho_{\delta}}
\newcommand{\ldelta}{l_{\delta}}
\newcommand{\pepsilon}{p_{\epsilon}}
\newcommand{\lepsilon}{l_{\epsilon}}
\newcommand{\ibeta}{i_{\beta}}
\newcommand{\iepsilon}{i_{\epsilon}}
\newcommand{\hepsilon}{h_{\epsilon}}
\newcommand{\hdelta}{h_{\delta}}
\newcommand{\pdelta}{p_{\delta}}

\begin{document}
\allowdisplaybreaks

\newcommand{\arXivNumber}{1903.10573}

\renewcommand{\PaperNumber}{069}

\FirstPageHeading

\ShortArticleName{Separability and Symmetry Operators for Painlev\'e Metrics}

\ArticleName{Separability and Symmetry Operators for Painlev\'e\\ Metrics and their Conformal Deformations}

\Author{Thierry DAUD\'E~$^\dag$, Niky KAMRAN~$^\ddag$ and Francois NICOLEAU~$^\S$}

\AuthorNameForHeading{T.~Daud\'e, N.~Kamran and F.~Nicoleau}

\Address{$^\dag$~D\'epartement de Math\'ematiques, UMR CNRS 8088, Universit\'e de Cergy-Pontoise,\\
\hphantom{$^\dag$}~95302 Cergy-Pontoise, France}
\EmailD{\href{mailto:thierry.daude@u-cergy.fr}{thierry.daude@u-cergy.fr}}

\Address{$^\ddag$~Department of Mathematics and Statistics, McGill University, \\
\hphantom{$^\ddag$}~Montreal, QC, H3A 2K6, Canada}
\EmailD{\href{mailto:niky.kamran@mcgill.ca}{niky.kamran@mcgill.ca}}

\Address{$^\S$~Laboratoire de Math\'ematiques Jean Leray, UMR CNRS 6629,\\
\hphantom{$^\S$}~2 Rue de la Houssini\`ere BP 92208, F-44322 Nantes Cedex 03, France}
\EmailD{\href{mailto:francois.nicoleau@univ-nantes.fr}{francois.nicoleau@univ-nantes.fr}}

\ArticleDates{Received March 27, 2019, in final form September 05, 2019; Published online September 16, 2019}

\Abstract{Painlev\'e metrics are a class of Riemannian metrics which generalize the well-known separable metrics of St\"ackel to the case in which the additive separation of variables for the Hamilton--Jacobi equation is achieved in terms of groups of independent variables rather than the complete orthogonal separation into ordinary differential equations which characterizes the St\"ackel case. Painlev\'e metrics in dimension $n$ thus admit in general only $r<n$ linearly independent Poisson-commuting quadratic first integrals of the geodesic flow, where $r$ denotes the number of groups of variables. Our goal in this paper is to carry out for Painlev\'e metrics the generalization of the analysis, which has been extensively performed in the St\"ackel case, of the relation between separation of variables for the Hamilton--Jacobi and Helmholtz equations, and of the connections between quadratic first integrals of the geodesic flow and symmetry operators for the Laplace--Beltrami operator. We thus obtain the generalization for Painlev\'e metrics of the Robertson separability conditions for the Helmholtz equation which are familiar from the St\"ackel case, and a formulation thereof in terms of the vanishing of the off-block diagonal components of the Ricci tensor, which generalizes the one obtained by Eisenhart for St\"ackel metrics. We also show that when the generalized Robertson conditions are satisfied, there exist $r<n$ linearly independent second-order differential operators which commute with the Laplace--Beltrami operator and which are mutually commuting. These operators admit the block-separable solutions of the Helmholtz equation as formal eigenfunctions, with the separation constants as eigenvalues. Finally, we study conformal deformations which are compatible with the separation into blocks of variables of the Helmholtz equation for Painlev\'e metrics, leading to solutions which are $R$-separable in blocks. The paper concludes with a set of open questions and perspectives.}

\Keywords{Painlev\'e metrics; Killing tensors; Helmholtz equation; $R$-separability; symmetry operators; Robertson conditions}

\Classification{53B21; 70H20; 81Q80}

\section{Introduction and statement of results}
Painlev\'e metrics \cite{Pain1897, Per1990} are a class of Riemannian metrics that provide a broad generalization of the well-known separable metrics of St\"ackel \cite{Eis1934, Sta1897} to the case in which the Hamilton--Jacobi equation for the geodesic flow can be additively separated into partial differential equations depending on groups of independent variables rather ordinary differential equations resulting from a complete orthogonal separation. In particular, while St\"ackel metrics in dimension $n$ admit $n$ linearly independent Poisson-commuting quadratic first integrals of the geodesic flow, Painlev\'e metrics in dimension $n$ will admit only $r<n$ such integrals in general, where $r$ denotes the number of groups of variables.

Our goal in this paper is to carry out the extension to Painlev\'e metrics of the well-known results \cite{Ben2016, BCR1-2002,BCR2-2002, Eis1934, KaMi1980, KaMi1981, KaMi1984} which relate in the St\"ackel case the additive separation of variables for the Hamilton--Jacobi equation to the multiplicative separation of variables for the Helmholtz equation, and the existence of quadratic first integrals of the geodesic flow to that of symmetry operators for the Laplace--Beltrami operator. We shall thus obtain the generalization to Painlev\'e metrics of the Robertson separability conditions~\cite{Rob1928} for the Helmholtz equation for St\"ackel metrics, and a formulation of these generalized Robertson conditions in terms of the vanishing of the off-block diagonal components of the Ricci tensor, thereby extending the classical result proved by Eisenhart \cite{Eis1934} for St\"ackel metrics. We shall also show that when the generalized Robertson conditions are satisfied, there exist $r<n$ linearly independent second-order differential operators which commute between themselves and with the Laplace--Beltrami operator. These operators will be shown to admit the block-separable solutions of the Helmholtz equation as formal eigenfunctions, with the separation constants arising from the separation into groups variables as eigenvalues. Finally, we shall also discuss conformal deformations of Painlev\'e metrics satisfying a further generalization of the Robertson conditions, which is compatible with the separation into blocks of variables of the Helmholtz equation, leading to solutions which are $R$-separable in blocks.

Before describing our results in further detail, we should remark that independently of the interest of Painlev\'e metrics from the point of view of separation of variables, a key motivation for our study lies in the goal of constructing geometric models of manifolds with boundary endowed with Painlev\'e metrics, with the goal of investigating the anisotropic Calder\'on problem in this class of geometries. Recall that the anisotropic Calder\'on problem consists in recovering the metric of a Riemannian manifold with boundary from the knowledge of the Dirichlet-to-Neumann map defined by the Laplace--Beltrami operator. The anisotropic Calder\'on problem is at the center of a great amount of current research activity. We refer to \cite{DKN2016,DKN2019a, GT2013, KKL2001, KS2014, LU2001, LeU1989, Sa2013, Uhl2009, Uhl2014} and the references therein for surveys of recent results on this problem. We have recently investigated the anisotropic Calder\'on problem at fixed energy for geometric models consisting of classes of St\"ackel manifolds with boundary, where the separation of variables for the Helmholtz equation allows the decomposition of the Dirichlet-to-Neumann map onto a~basis of joint eigenfunctions of the symmetry operators resulting from the complete separation of variables, enabling us to obtain a series of uniqueness and non-uniqueness results for the Calder\'on problem, with no a-priori assumptions of analyticity, or on the existence of isometries \cite{DKN2015, DKN2016, DKN2017, DKN2018a, DKN2019b, Gob2016}. In the case of Painlev\'e metrics, the separation of the Helmholtz into groups of variables and the concomitant families of commuting symmetry operators admitted by these metrics will serve as an effective starting point for the investigation of the anisotropic Calder\'on problem in this more general setting.

In order to put the results of the present paper in context, we first briefly recall some well-known definitions and results pertaining to St\"ackel metrics and their separability properties. Throughout the paper, we shall assume for simplicity that the manifolds, metrics and maps being considered are smooth, although many of the results that we quote or obtain can be shown to hold under weaker differentiability properties. Recall \cite{BCR1-2002, Eis1934, KaMi1980, Sta1897} that a {\emph {St\"ackel metric}} on an $n$-dimensional manifold $M$ is a Riemannian metric $g$ for which there exist local coordinates $\big(x^{1},\dots,x^{n}\big)$ in which the metric has the expression
\begin{gather}\label{Stackelform}
{\rm d}s^{2}=g_{ij}{\rm d}x^{i}{\rm d}x^{j}=\frac{\det S}{s^{11}}\big({\rm d}x^{1}\big)^{2}+\cdots+\frac{\det S}{s^{n1}}\big({\rm d}x^{n}\big)^{2},
\end{gather}
where $S$ is a {\emph{St\"ackel matrix}}, that is a non-singular $n\times n$ matrix $S=(s_{ij})$ of the form
\begin{gather}\label{StackelMatrix}
S=\left(
\begin{matrix}
s_{11}\big(x^{1}\big)&\dots&s_{1n}\big(x^{1}\big) \\
\vdots& & \vdots \\
s_{n1}\big(x^{n}\big)&\dots&s_{nn}\big(x^{n}\big)
\end{matrix}
\right),
\end{gather}
and $s^{ij}$ denotes the cofactor of the component $s_{ij}$ of the matrix~$S$. St\"ackel matrices thus have the property that for each $1\leq i \leq n$, their $i$-th row depends only on the $i$-th local coordinate~$x^{i}$, and that the cofactor $s^{ij}$ is independent of the $i$-th local coordinate $x^{i}$. Furthermore, the diagonal components of the St\"ackel metric~(\ref{Stackelform}) are given by the inverses of the entries of the first row of the inverse St\"ackel matrix $A=S^{-1}$.

The importance of St\"ackel metrics stems from the fact that they constitute the most general class of Riemannian metrics admitting orthogonal coordinates for which the Hamilton--Jacobi equation
\begin{gather}\label{HJ}
g^{ij}\partial_{i}W\partial_{j}W=E,
\end{gather}
for the geodesic flow of $(M,g)$, where $E$ denotes a positive real constant, admits a complete integral obtained by additive separation of variables into ordinary differential equations. It is useful at this stage to recall that a complete integral of (\ref{HJ}) is defined as a parametrized family of solutions
\begin{gather}\label{complint}
W=W\big(x^{1},\dots,x^{n};a_{1},\dots,a_{n}\big),\qquad a_{1}:=E,
\end{gather}
defined over the domain $U\subset M$ of the local coordinates $\big(x^{1},\dots,x^{n}\big)$ and depending smoothly on $n$ parameters $(a_{1},\dots,a_{n})$ defined on an open subset $A\subset {\mathbb{R}}^{n}$, such that the rank condition
\begin{gather*}
\det \left(\frac{\partial^{2}W}{\partial x^{i}\partial a_{j}}\right) \neq 0,
\end{gather*}
is satisfied throughout the open set $U\times A$.

It is easily verified that the Hamilton--Jacobi equation (\ref{HJ}) will admit an additively separable complete integral $W\big(x^{1},\dots, x^{n};a_{1},\dots a_{n}\big)$ of the form
\[
W=W_{1}\big(x^{1};a_{1},\dots, a_{n}\big)+\cdots+W_{n}\big(x^{n};a_{1},\dots, a_{n}\big),
\]
if and only if the summands $W_{\alpha}$ satisfy the set of separated ordinary differential equations given by
\[
\left(\frac{{\rm d}W_{i}}{{\rm d}x^{i}}\right)^{2}=\sum_{j=1}^{n}s_{ij}\big(x^{i}\big)a_{j}.
\]
The $n$ parameters $(a_{1},\dots , a_{n})$ appearing in the expression of the additively separable complete integral~(\ref{complint}) thus correspond to the separation constants arising from the complete separation of variables of the Hamilton--Jacobi equation into ordinary differential equations.
One of the key consequences of this complete separation of variables property is that the geodesic flow of an $n$-dimensional St\"ackel metric admits a linearly independent set of $n-1$ mutually Poisson-commuting quadratic first integrals, given by
\begin{gather}\label{FirstInt}
K_{(l)} = K_{(l)}^{ij} p_i p_j = \sum_{j=1}^{n}a_{lj} p_{j}^{2}, \qquad 2\leq l \leq n
\end{gather}
with
\[ \big\{K_{(l)},K_{(m)}\big\}=0,\qquad \text{for} \quad 1\leq l,m \leq n,
\]
where $A=(a_{ij})$ denotes as above the inverse of the St\"ackel matrix $S$ given by (\ref{StackelMatrix}). Note that with the notations of (\ref{FirstInt}), we have $K_{(1)}=H$, where
\[
H=g^{ij}p_{i}p_{j},
\]
denotes the Hamiltonian for the geodesic flow.

 A question closely related to the additive separation of the Hamilton--Jacobi equation is that of the complete multiplicative separation of the Helmholtz equation
 \begin{gather}\label{Helm}
-\Delta_{g} u =\lambda u,
\end{gather}
where
\[
\Delta_{g} = \frac{1}{\sqrt{|g|}}\partial_{i}\Big(\sqrt{|g|}g^{ij}\partial_{j}\Big),\qquad |g|=\det(g_{ij}),
\]
denotes the Laplace--Beltrami operator on~$(M,g)$ and $\lambda$ denotes a non vanishing real constant, into ordinary differential equations. By complete multiplicative separation, we mean, follo\-wing~\cite{BCR1-2002}, the existence of a parametrized family of solutions $u$ defined on a domain $U\subset M$ with local coordinates $\big(x^{1},\dots, x^{n}\big)$ of the form
\[
u\big(x^{1},\dots, x^{n};a_{1},\dots, a_{2n}\big)=\prod_{i=1}^{n}u_{i}\big(x^{1},\dots, x^{n};a_{1},\dots, a_{2n}\big),
\]
depending smoothly on $2n$ parameters $(a_{1},\dots, a_{2n})$ defined on open subset $A\subset {\mathbb R}^{2n}$, satisfying the rank condition
\begin{gather*}
\det \left(
\begin{matrix}
\dfrac{\partial v_{1}}{\partial a_{1}}&\dots&\dfrac{\partial v_{1}}{\partial a_{2n}} \\
\vdots& & \vdots \\
\dfrac{\partial v_{n}}{\partial a_{1}}&\dots&\dfrac {\partial v_{n}}{\partial a_{2n}}\vspace{1mm}\\
\dfrac{\partial w_{1}}{\partial a_{1}}&\dots&\dfrac{\partial w_{1}}{\partial a_{2n}} \\
\vdots& & \vdots \\
\dfrac{\partial w_{n}}{\partial a_{1}}&\dots&\dfrac {\partial w_{n}}{\partial a_{2n}}
\end{matrix}
\right) \neq 0,
\end{gather*}
at every point of $U\times A$, where
\[
v_{i}=\frac{u_{i}'}{u_{i}} \qquad w_{i}=\frac{u_{i}''}{u_{i}}.
\]
This separation requires that additional conditions, known as the {\emph{Robertson conditions}}, and given by
\[
\partial_{i}\Gamma_{j}=0,\qquad 1\leq i\neq j\leq n,
\]
where
\[
\Gamma_{i}=-\partial_{i}\left[\log \frac{(\det S)^{\frac{n}{2}-1}s^{i1}}{\big(s^{11}\cdots s^{n1}\big)^{\frac{1}{2}}}\right],
\]
be satisfied. We refer to \cite{BCR1-2002,BCR2-2002,Eis1934,KaMi1980,KaMi1981,Rob1928} for detailed proofs of the fact that the Robertson conditions are necessary and sufficient for complete multiplicative separation of the Helmholtz equation. The Robertson conditions were interpreted by Eisenhart \cite{Eis1934} in terms of the vanishing of the off-diagonal components of the Ricci tensor of the underlying St\"ackel metric, that is
\begin{gather}\label{classicalRobertson}
R_{ij}=0\qquad \text{for}\qquad 1\leq i\neq j\leq n.
\end{gather}
When the Robertson conditions are satisfied, the Poisson-commuting quadratic first integrals of the geodesic flow give rise to $n-1$ linearly independent second-order differential operators which commute with the Laplace--Beltrami operator and also commute pairwise. Rewriting the quadratic first integrals $K_{(l)}$ defined by (\ref{FirstInt}) in the form
\[
K_{(l)}=K_{(l)}^{ij}p_{i}p_{j},
\]
these commuting operators, denoted by $\Delta_{K_{(l)}}$, are of the form
\[
\Delta_{K_{(l)}}=\nabla_{i}\big(K_{(l)}^{ij}\nabla_{j}\big),\qquad 2\leq l \leq n,
\]
where $\nabla_i$ denotes the Levi-Civita connection on $(M,g)$. These operators, which are often referred to as {\emph{symmetry operators}}, admit the separable solutions of the Helmholtz equation as (formal) eigenfunctions. We will not give any further details on symmetry operators at this stage, nor shall we say anything about the proofs of the results we have just recalled since we shall shortly state and prove generalizations of these to the case of Painlev\'e metrics, which admit all St\"ackel metrics as a special case.

We conclude these preliminaries by remarking that the above setting may be expanded significantly by considering conformal deformations of St\"ackel metrics which are compatible with the complete separation of the Helmholtz equation into ordinary differential equations, thus giving rise to the more general notion of {\emph{$R$-separability}} for the Helmholtz equation. Again, we shall not give any additional details on these topics at this stage since conformal deformations and $R$-separability will be studied in the remainder of this paper in the more general setting of Painlev\'e metrics. We refer to \cite{BCR1-2002, BCR2-2002, BCR2005, CR2006,CR2007, KaMi1980, KaMi1981, KaMi1984} for lucid accounts of the key results on the separability and $R$-separability properties of St\"ackel metrics, their connection to Killing tensors, quadratic first integrals of the geodesic flow and symmetry operators for the Laplace--Beltrami operator. We also refer to \cite{Ben2016, Mi1988} for recent surveys on separability on Riemannian manifolds and to \cite{Car1977} for a penetrating analysis of the relations between quadratic first integrals of the geodesic flow, symmetry operators and conserved currents, in the general setting of Riemannian or pseudo-Riemannian manifolds.

With these preliminaries at hand, we are now ready to introduce the class of Painlev\'e met\-rics~\cite{Pain1897, Per1990}. As stated above, Painlev\'e metrics arise as a natural generalization of St\"ackel metrics to the case in which one no longer seeks complete separation of the Hamilton--Jacobi equation into ordinary differential equations, but rather separation into partial differential equations involving groups of variables. The separable coordinates admitted by Painlev\'e metrics are thus generally {\emph {not orthogonal}}, although they are orthogonal with respect to groups of variables. Let us recall that our goal in this paper is to carry out for Painlev\'e metrics the analogue of the separability and $R$-separability analyses of the Helmholtz equation which has been extensively worked out for St\"ackel metrics in \cite{BCR1-2002,BCR2-2002,CR2006,KaMi1980,KaMi1981, KaMi1984}, and to show that the separability into groups of variables gives rise to vector spaces of mutually commuting symmetry operators for the Laplace--Beltrami operator, the dimension of which is determined by the number of groups of variables. In particular, we will generalize to the case of Painlev\'e metrics the Robertson conditions and the characterization thereof in terms of the Ricci tensor. We now proceed to define the class of Painlev\'e metrics along lines similar to the ones used above for St\"ackel metrics.

Let $(M,g)$ be an $n$-dimensional Riemannian manifold and let ${\bf x}=\big(x^{1},\dots , x^{n}\big)$ denote a set of local coordinates on $M$. We shall consider partitions of ${\bf x}$ into $r\geq 2$ groups of local coordinates,
\[
{\bf x}=\big({{\bf x}^{1}},\dots,{{\bf x}^{r}}\big),
\]
where
\[
{\bf {x}}^{\alpha}=\big(x^{i_{\alpha}}\big),\qquad 1_{\alpha}\leq i_{\alpha} \leq l_{\alpha},\qquad \text{and}\qquad \sum_{\alpha=1}^{r}||l_{\alpha}||=n,
\]
where $||l_{\alpha}||$ denotes the number of local coordinates in the group of label $\alpha$. Latin indices $1\leq i,j,\ldots \leq n$ will be used to label the local coordinates on $M$, greek indices $\alpha,\beta,\dots$ to label the $r$ groups of local coordinates, and hybrid indices $i_{\alpha}$, $1_{\alpha} \leq i_{\alpha}\leq l_{\alpha}$ to denote the local coordinates within the group $\bf{x^{\alpha}}$. Unless there is an ambiguity in the notation being used, in which case we will write out the summation signs explicitly, we shall apply the summation convention with the above range of indices.
A {\em generalized St\"ackel matrix} is a non-singular $r\times r$ matrix-valued function $S$ on $M$ of the form
\begin{gather}\label{genStackelmat}
S=\left(
\begin{matrix}
s_{11}\big({{\bf x}^{1}}\big)&\dots&s_{1r}\big({{\bf x}^{1}}\big) \\
\vdots& & \vdots \\
s_{r1}\big({{\bf x}^{r}}\big)&\dots&s_{rr}\big({\bf x}^{r}\big)
\end{matrix}
\right),
\end{gather}
where for each $1\leq \alpha \leq r$,
\[
{\bf x}^{\alpha}=\big(x^{i_{\alpha}}\big),\qquad 1_{\alpha}\leq i_{\alpha} \leq l_{\alpha}.
\]
Let $s^{\alpha \beta}$ denote the cofactor of the component $s_{\alpha \beta}$ of $S$. We note that the cofactor $s^{\beta \gamma}$ will not depend on the group of variables ${\bf x}^{\beta}=\big(x^{i_{\beta}}\big)$, $1_{\beta}\leq i_{\beta} \leq l_{\beta}$. Moreover we shall assume that
\begin{gather} \label{StackelPositiveness}
\frac{\det S}{s^{\alpha 1}} > 0, \qquad \forall\, 1\leq \alpha \leq r ,
\end{gather}
in order to work with {\emph{Riemannian}} Painlev\'e metrics.

\begin{Definition} \label{Painleve}
Let $S$ be a generalized St\"ackel matrix satisfying (\ref{StackelPositiveness}). A Painlev\'e metric is a~Riemannian metric~$g$ for which there exist local coordinates ${\bf x}=\big({\bf x}^{1},\dots,{\bf x}^{r}\big)$ such that
\begin{gather}\label{Painleveform}
{\rm d}s^{2}=g_{ij}{\rm d}x^{i}{\rm d}x^{j}=\frac{\det S}{s^{11}} G_1 +\cdots+\frac{\det S}{s^{r1}} G_r,
\end{gather}
where each of the quadratic differential forms
\begin{gather}\label{blockmetric}
G_{\alpha}=G_{\alpha}({{\bf x}^{\alpha}})=\sum_{i_{\alpha}=1_{\alpha}}^{l_{\alpha}}\sum_{j_{\alpha}=1_{\alpha}}^{l_{\alpha}}G_{\alpha}\big({{\bf x}^{\alpha}}\big)_{i_{\alpha}j_{\alpha}}{\rm d}x^{i_{\alpha}}{\rm d}x^{j_{\alpha}},\qquad 1\leq \alpha \leq r,
\end{gather}
is positive-definite in its arguments and depends only on the group of variables ${\bf x}^{\alpha}$.
\end{Definition}
We may thus write the metric (\ref{Painleveform}) in block-diagonal form as
\[
{\rm d}s^{2}=\sum_{\alpha =1}^{r}\sum_{\ialpha = 1_{\alpha}}^{\lalpha}\sum_{\jalpha = 1_{\alpha}}^{\lalpha}(g_{\alpha})_{\ialpha \jalpha}{\rm d}x^{\ialpha}{\rm d}x^{\jalpha}=\sum_{\alpha =1}^{r} \frac{\det S}{s^{\alpha 1}} \sum_{\ialpha = 1_{\alpha}}^{\lalpha}\sum_{\jalpha = 1_{\alpha}}^{\lalpha}(G_{\alpha})_{i_{\alpha}j_{\alpha}}{\rm d}x^{i_{\alpha}}{\rm d}x^{j_{\alpha}}.
\]
It is important to note that even though Painlev\'e metrics (\ref{Painleveform}) are block-diagonal, and each quadratic differential form (\ref{blockmetric}) defines a Riemannian metric on the submanifolds defined by the level sets ${\bf x}^{\beta}={\bf{c}}^{\beta}$, $\beta\neq \alpha$, {\emph {a Painlev\'e metric is generally not a direct sum of Riemannian metrics, nor a warped product, except for special non-generic cases}}. We also note that Painlev\'e metrics of semi-Riemannian (and in particular Lorentzian) signature can readily be defined by modifying the requirement that each of the quadratic differential forms $G_{\alpha}$ given by (\ref{blockmetric}) be positive-definite to one in which $G_{{\alpha}}$ is assumed to be of signature $(p_{\alpha},q_{\alpha})$ with $p_{\alpha}+q_{\alpha}=l_{\alpha}$. Finally we also remark that the Painlev\'e form (\ref{Painleveform}) is obviously invariant under smooth and invertible changes of coordinates of the form ${\tilde{\bf x}}^{\alpha}={\bf{f}}^{\alpha}({\bf x}^{\alpha})$, where $1\leq \alpha \leq r$.

Let us call \emph{block orthogonal coordinates} a system of coordinates $({\bf x}^{\alpha})$ such that the metric $g$ has the form
\begin{gather*} 
 g=\sum_{\alpha =1}^{r} c_\alpha G_\alpha = \sum_{\alpha =1}^{r} c_\alpha \sum_{\ialpha = 1_{\alpha}}^{\lalpha}\sum_{\jalpha = 1_{\alpha}}^{\lalpha}(G_{\alpha})_{i_{\alpha}j_{\alpha}}{\rm d}x^{i_{\alpha}}{\rm d}x^{j_{\alpha}},
\end{gather*}
where $c_\alpha$ are non-vanishing scalar functions on $M$ and the metrics $G_\alpha$ are given by (\ref{blockmetric}).
In analogy with the St\"ackel case, we have

\begin{Proposition} \label{HJSeparability-Painleve}A metric $g$ is of the Painlev\'e form \eqref{Painleveform} if and only if there exist block orthogonal coordinates such that the Hamilton--Jacobi equation
\[
g^{ij}\partial_{i}W\partial_{j}W=E,
\]
admits a parametrized family of solutions which is sum-separable into groups of variables, of the form
\begin{gather}\label{Painsepform}
W=W_{1}\big({\bf {x}}^{1};{{a}}_{1},\dots ,{{a}}_{r}\big)+\cdots+W_{r}\big({\bf {x}}^{r};{{{a}}_{1},\dots, {a}}_{r}\big),
\end{gather}
depending smoothly on $r$ arbitrary real constants $({{a}}_{1}:=E,a_{2},\dots ,{{a}}_{r})$ defined on an open subset $A\subset \mathbb{R}^{r}$, and satisfying the rank condition
\begin{gather}\label{rankPainHJ}
\det \left(\sum_{\ialpha =1_{\alpha}}^{\lalpha}\sum_{\jalpha =1_{\alpha}}^{\lalpha}\big(G^{\alpha}\big)^{\ialpha \jalpha}\left(\frac{\partial W_{\alpha}}{\partial x^{\ialpha}}\right)\left( \frac{\partial^{2}W_{\alpha}}{\partial a_{\gamma}\partial x^{\jalpha}}\right)\right)\neq 0,
\end{gather}
where
\[
G^{\alpha} = (G_{\alpha})^{-1 }.
\]
\end{Proposition}
This Proposition will be proved in Section \ref{PainHJSection} as well as other (intrinsic) characterizations of Painlev\'e metrics.

In further analogy with the St\"ackel case, we now recall that Painlev\'e metrics admit $r$ linearly independent quadratic first integrals of the geodesic flow which are Poisson commuting. Indeed, the summands $W_{\alpha}$ appearing in (\ref{Painsepform}) satisfy the following set of first-order PDEs \cite{Per1990}
\begin{gather}
{\mathcal{F}}_{1}\left({\bf {x}}^{1},\frac{\partial W_{1}}{\partial {\bf x}^{1}}\right) =a_{1}s_{11}\big({\bf {x}}^{1}\big)+\cdots +a_{r}s_{r1}\big({\bf {x}}^{1}\big),\nonumber\\ \quad \vdots \nonumber\\ {\mathcal{F}}_{r}\left({\bf {x}}^{r},\frac{\partial W_{r}}{\partial {\bf x}^{r}}\right) =a_{1}s_{r1}\big({\bf {x}}^{r}\big)+\cdots +a_{r}s_{rr}\big({\bf {x}}^{r}\big),\label{HJsep}
\end{gather}
where
\begin{gather}\label{F}
{\mathcal{F}}_{\alpha}=\big(G^{\alpha}\big)^{i_{\alpha}j_{\alpha}}\big({\bf {x}}^{\alpha}\big)p_{i_{\alpha}}p_{j_{\alpha}},\qquad 1\leq \alpha \leq r,
\end{gather}
and where the $(a_{\alpha})$ are arbitrary real separation constants.
It follows now directly from the separated equations (\ref{HJsep}) and from the fact that the generalized St\"ackel matrix $S$ is non-singular that
one obtains $r$ linearly independent Poisson-commuting quadratic first integrals~$K_{(\alpha)}$ of the geodesic flow by solving for the $r$ separation constants $(a_{\alpha})$ from the separated equations (\ref{HJsep}). These are explicitly given by (see~\cite{Per1990})
\[
K_{(\alpha)}=\sum_{\beta =1}^{r}\big(S^{-1}\big)_{\alpha \beta}\mathcal{F}_{\beta},\qquad 1\leq \beta \leq r,
\]
or equivalently
\[
K_{(\alpha)}=K_{(\alpha)}^{ij}p_{i}p_{j},
\]
where for each $1\leq \alpha \leq r$, $\big(K_{(\alpha)}^{ij}\big)$ is a symmetric block-diagonal tensor defined by
\begin{gather}\label{Killingtensor}
K_{(\alpha)}^{i_{\beta}j_{\beta}}=\frac{s^{\beta \alpha}}{\det S}\big(G^{\beta}\big)^{i_{\beta}j_{\beta}},\qquad K_{(\alpha)}^{i_{\beta}j_{\gamma}}=0\qquad \text{for all}\qquad 1 \leq \beta \neq \gamma \leq r.
\end{gather}
These quadratic first integrals satisfy
\begin{gather}\label{PoissonCommute}
\big\{K_{(\alpha)}, K_{(\beta)}\big\}=0, \qquad 1\leq \alpha,\beta \leq r, \qquad \text{where}\qquad K_{(1)}=H.
\end{gather}
The condition $\{K_{(\alpha)},H\}=0$ is equivalent to $\big(K_{(\alpha)}^{ij}\big)$ being a symmetric {\emph{Killing tensor}}, that is
\[
\nabla_{i}K_{(\alpha) jk}+\nabla_{j}K_{(\alpha) ki}+\nabla_{k}K_{(\alpha) ij}=0.
\]
The commutation relations (\ref{PoissonCommute}) are thus equivalent to the vanishing of the Schouten brackets of the pairs of Killing tensors $(K_{(\alpha) ij}),(K_{(\beta) ij})$.

There exist a few classical examples of Painlev\'e metrics in the litterature. They include for instance the {\emph{di Pirro metrics}} \cite{diPirro1896, Per1990}, for which the Hamiltonian of the geodesic flow is of the form
\begin{gather}\label{diPirro}
H=g^{ij}p_{i}p_{j}=\big(c_{12}\big(x^{1},x^{2}\big)+c_{3}\big(x^{3}\big)\big)^{-1}\big[a_{1}\big(x^{1},x^{2}\big)p_{1}^{2}+a_{2}\big(x^{1},x^{2}\big)p_{2}^{2} +a_{3}\big(x^{3}\big)p_{3}^{2}\big].
\end{gather}
It may indeed be verified directly that the function
\[
K=\big(c_{12}\big(x^{1},x^{2}\big)+c_{3}\big(x^{3}\big)\big)^{-1}\big[c_{3}\big(x^{3}\big)(a_{1}\big(x^{1},x^{2}\big)p_{1}^{2} +a_{2}\big(x^{1},x^{2}\big)p_{2}^{2})-c_{12}\big(x^{1},x^{2}\big)a_{3}\big(x^{3}\big)p_{3}^{2}\big].
\]
Poisson-commutes with $H$, and thus defines a Killing tensor, which together with the metric tensor generates a maximal linearly independent set of Killing tensors for generic choices of the metric functions $c_{12}$, $a_{1}$, $a_{2}$, $a_{3}$, $c_{3}$ in~(\ref{diPirro}). Painlev\'e metrics also appear in the context of geodesically equivalent metrics as metrics admitting projective symmetries, see \cite{Top1999, Top2002}, and also as instances of 4-dimensional Lorentzian metrics admitting a Killing tensor \cite{HaMa1978} (see Section~\ref{Perspectives} for further remarks on the latter point). Finally, we mention the recent paper by Chanu and Rastelli~\cite{CR2019} that provides a classification of Painlev\'e metrics with vanishing Riemann tensor in dimension~$3$,  i.e., in~$\mathbb{E}_3$. We will give some examples of Painlev\'e metrics in all dimensions satisfying the generalized Robertson conditions (see below) as well as a catalogue of such metrics in dimensions $2$, $3$, $4$ in Section \ref{2}.

As we stated above, our main goal in this paper is to investigate for the class of Painlev\'e metrics the closely related question of product separability for the Helmholtz equation (\ref{Helm}), and the relationship between quadratic first integrals of the geodesic flow and symmetry operators for the Laplace--Beltrami equation. The Laplace--Beltrami operator for a Painlev\'e metric $g$ given by~(\ref{Painleveform}) can be expressed in terms of the generalized St\"ackel matrix $S$ and the Laplace--Beltrami operators for the $r$ Riemannian metrics $G_{\beta}$, $1\leq \beta \leq r$, defined by~(\ref{blockmetric}), corresponding to the blocks of variables ${\bf x}^{\beta}$, $1\leq \beta \leq r$. We have
\begin{gather}
\Delta_{g} u =\sum_{\beta=1}^{r}\Bigg[\left(\frac{s^{\beta 1}}{\det S}\right)\Bigg\{\Delta_{G_{\beta}}u\nonumber\\
\hphantom{\Delta_{g} u =}{}+\sum_{i_{\beta}=1_{\beta}}^{l_{\beta}}\sum_{{j}_{\beta}=1_{\beta}}^{l_{\beta}}\big(G^{\beta}\big)^{i_{\beta} j_{\beta}}\left[\partial_{i_{\beta}}\left(\log \frac{(\det S)^{\frac{n}{2}-1}s^{\beta 1}}{\big(s^{11}\big)^{\frac{l_1}{2}} \cdots \big(s^{r1}\big)^{\frac{l_r}{2}}}\right)\right]\partial_{j_{\beta}}u\Bigg\}\Bigg],\label{Laplacian}
\end{gather}
where $\Delta_{G_{\beta}}$ denotes the Laplace--Beltrami operator for the Riemannian metric $G_{\beta}$, that is
\begin{gather*}
\Delta_{G_{\beta}} = \sum_{i_{\beta}=1}^{\lbeta}\sum_{j_{\beta}=1}^{\lbeta}\frac{1}{\sqrt{|G_{\beta}|}}\partial_{i_{\beta}}\Big(\sqrt{|G_{\beta}|}
\big(G^{\beta}\big)^{i_{\beta}j_{\beta}}\partial_{j_{\beta}}\Big),\qquad |G_{\beta}|=\det((G_{\beta})_{i_{\beta}j_{\beta}}).
\end{gather*}

We now state our main results, the proofs of which will be given in Section \ref{ProofsofMainTheorems}. We first define the \emph{generalized Robertson conditions}, in analogy with the classical Robertson conditions (\ref{classicalRobertson}) for St\"ackel metrics.

\begin{Definition} \label{defiRC}
A Painlev\'e metric $g$ is said to satisfy the generalized Robertson conditions if and only if the differential conditions
	\begin{gather}\label{GenRobertson}
\partial_{j_{\alpha}}\gamma_{i_{\beta}}=0,\quad \textrm{for all} \quad 1\leq \alpha \neq \beta \leq r, \quad 1_{\alpha}\leq i_{\alpha} \leq l_{\alpha},\quad 1_{\beta}\leq i_{\beta} \leq l_{\beta},
\end{gather}
where
\[
\gamma_{i_{\beta}}=-\partial_{i_{\beta}}\left[\log \frac{(\det S)^{\frac{n}{2}-1}s^{\beta 1}}{\big(s^{11}\big)^{\frac{l_1}{2}} \cdots \big(s^{r1}\big)^{\frac{l_r}{2}}}\right],
\]
are satisfied.
\end{Definition}
The generalized Robertson conditions (\ref{GenRobertson}) imply that
\begin{gather}\label{GenRobertsonsimpl}
\partial_{\ialpha}\gamma^{\jbeta}=0, \quad \textrm{for all} \quad 1\leq \alpha \neq \beta \leq r,\\
\label{gammaup}
\gamma^{j_{\beta}}:= \sum_{i_{\beta}=1_{\beta}}^{l_{\beta}}\big(G^{\beta}\big)^{i_{\beta}j_{\beta}} \gamma_{i_{\beta}}.
\end{gather}
We shall be working with both the forms (\ref{GenRobertson}) and (\ref{GenRobertsonsimpl}) of these conditions. Note that if the Robertson conditions hold, then the positive Laplace--Beltrami operator can be written in a synthetic form as
\begin{gather*}
 -\Delta_{g} u  =   \sum_{\beta=1}^{r} \left(\frac{s^{\beta 1}}{\det S}\right)\left\{-\Delta_{G_{\beta}}u + \sum_{{j}_{\beta}=1_{\beta}}^{l_{\beta}} \gamma^{j_{\beta}} \partial_{j_{\beta}}u\right\}
   =   \sum_{\beta=1}^{r} \left(\frac{s^{\beta 1}}{\det S}\right) B_\beta u , \nonumber
\end{gather*}
where the operators $B_\beta = -\Delta_{G_{\beta}} + \sum\limits_{{j}_{\beta}=1_{\beta}}^{l_{\beta}} \gamma^{j_{\beta}} \partial_{j_{\beta}}$ only depend on the group of variables $\textbf{x}^\beta$.

As will be seen in Section \ref{PainHelmSection}, the generalization of the notion of complete multiplicative separation for the Helmholtz equation to the case of separation in terms of groups of variables is given by considering a parametrized family of product-separable solutions of the form
\begin{gather} \label{SeparabilitySchro}
u=\prod_{\beta =1}^{r}u_{\beta}\big({\bf x}^{\beta};a_{1},\dots,a_{r}\big),
\end{gather}
satisfying the rank condition
\begin{gather} \label{RankPainSchro}
 \det  \left( \partial_{a_\alpha} \left( \frac{B_\beta u_\beta}{u_\beta} \right) \right) \ne 0,
\end{gather}
where we assume that $u_{\beta}\neq 0$.

Our first result states the separability conditions for the Helmholtz equation, and gives their interpretation in terms of the vanishing of the off-block diagonal components of the Ricci tensor:

\begin{Theorem}\label{maintheorem1}\quad
\begin{enumerate}\itemsep=0pt
\item[$1)$] Given a Painlev\'e metric $g$ of the form \eqref{Painleveform} satisfying the Robertson conditions \eqref{GenRobertson}, the Helmholtz equation
\begin{gather}\label{HelmPain}
-\Delta_{g} u =\lambda u,
\end{gather}
where $\Delta_{g}$ denotes the Laplace--Beltrami operator \eqref{Laplacian} admits a solution that is product-separab\-le in the $r$ groups of variables $\big({{\bf x}^{1}},\dots,{{\bf x}^{r}}\big)$,
\begin{gather}\label{sepform}
u=\prod_{\beta =1}^{r}u_{\beta}\big({\bf x}^{\beta};a_{1}:=\lambda,a_{2},\dots,a_{r}\big),
\end{gather}
and satisfies the rank condition \eqref{RankPainSchro}.

\item[$2)$] The conditions \eqref{GenRobertson} may be written equivalently in terms of the Ricci tensor of the Painlev\'e metric \eqref{Painleveform} as
\begin{gather}\label{Ricciform}
R_{j_{\alpha}k_{\beta}}=0, \quad {\mbox{for all}} \quad 1\leq \alpha \neq \beta \leq r, \quad {\mbox{and}} \quad 1_{\alpha}\leq j_{\alpha} \leq l_{\alpha}, \quad 1_{\beta}\leq k_{\beta}\leq l_{\beta}.
\end{gather}
\end{enumerate}
\end{Theorem}

Our next result shows that the Laplace--Beltrami operator for a Painlev\'e metric satisfying the generalized Robertson conditions admits $r$ linearly independent mutually commuting symmetry operators:

\begin{Theorem}\label{maintheorem2} Consider a Painlev\'e metric \eqref{Painleveform} for which the generalized Robertson conditions~\eqref{GenRobertson}, which imply the separability of the Helmholtz equation, are satisfied. Then the operators~$\Delta_{K_{\alpha}}$ defined for $2\leq \alpha \leq r$ by
\begin{gather}\label{KLaplacian}
\Delta_{K_{(\alpha)}}=\nabla_{i}(K_{(\alpha)}^{ij}\nabla_{j})=\sum_{\gamma=1}^{r}\sum_{i_{\gamma} =1_{\gamma}}^{l_{\gamma}}\sum_{j_{\gamma}=1_{\gamma}}^{l_{\gamma}}\nabla_{i_{\gamma}}\big(K_{(\alpha)}^{i_{\gamma}j_{\gamma}}\nabla_{j_{\gamma}}\big),
\end{gather}
where $K_{(\alpha)}$ is defined by~\eqref{Killingtensor},
commute with the Laplace--Beltrami opera\-tor~$\Delta_{g}$ and pairwise commute
\begin{gather}\label{commrelKill}
\big[\Delta_{K_{(\alpha)}},\Delta_{g}\big]=0,\qquad \big[\Delta_{K_{(\alpha)}},\Delta_{K_{(\beta)}}\big]=0,\qquad 2\leq \alpha,\beta \leq r,
\end{gather}
and admit the separable solutions \eqref{sepform} as formal eigenfunctions with the separation cons\-tants~$a_{\alpha}$ arising from the separation of variables as eigenvalues,
\[
\Delta_{K_{(\alpha)}}u=a_{\alpha}u,\qquad 2\leq \alpha \leq r.
\]
\end{Theorem}
Our final result shows that the above framework can be expanded just as in the St\"ackel case by considering conformal deformations of the Painlev\'e metrics (\ref{Painleveform}) which are compatible with the separation of the Helmholtz equation into groups of variables. This corresponds to a generalization of the important notion of $R$-separability \cite{BCR2005, CR2006, KaMi1984} to the context of Painlev\'e metrics.

Let us first recall that upon a conformal rescaling of the metric given by
\begin{gather}\label{confresc}
g\mapsto c^{4}g,
\end{gather}
where $c$ denotes a smooth positive function, the Laplace--Beltrami operator $\Delta_{g}$ obeys the transformation law
\begin{gather}\label{Yamabe1}
\Delta_{c^{4}g}=c^{-(n+2)}\big(\Delta_{g}-q_{c,g}\big)c^{n-2},
\end{gather}
where
\begin{gather}\label{Yamabe2}
q_{c,g}=c^{-n+2}\Delta_{g}c^{n-2}.
\end{gather}
Thus, letting
\begin{gather}\label{rescaledu}
v=c^{n-2}u
\end{gather}
and using the expression (\ref{Laplacian}) of the Laplace--Beltrami operator for a Painlev\'e metric $g$, the Helmholtz equation
\begin{gather}\label{confHelmholtz}
-\Delta_{c^{4}g}u=\lambda u,
\end{gather}
takes the form
\begin{gather}\label{Yamabe3}
\left[\sum_{\beta=1}^{r}\frac{s^{\beta 1}}{\det S}\left(-\Delta_{G_{\beta}}+\sum_{j_{\beta=1_{\beta}}}^{l_{\beta}}\gamma^{j_{\beta}}\partial_{j_{\beta}}\right)+q_{c,g}-\lambda c^{4}\right]v=0.
\end{gather}
We have
\begin{Theorem}\label{maintheorem3}
Let $g$ be a Painlev\'e metric. Suppose furthermore that $g$ is conformally rescaled by a factor $c^{4}$ as in \eqref{confresc}, where $c$ is chosen so as to satisfy the non-linear PDE
\begin{gather}\label{eqnc}
\Delta_{g}c^{n-2}-\lambda c^{n+2}-\left( a_1 + \sum_{\beta =1}^{r}\frac{s^{\beta 1}}{\det S}(P_{\beta}-\phi_{\beta}) \right) c^{n-2}=0,
\end{gather}
where
\begin{gather*}
P_{\beta}:=-\frac{1}{2}\partial_{j_{\beta}}\gamma^{j_{\beta}}-\frac{1}{4}\gamma^{j_{\beta}}\partial_{j_{\beta}}\log |G_{\beta}|+\frac{1}{4}(G_{\beta})_{i_{\beta}j_{\beta}}\gamma^{i_{\beta}}\gamma^{j_{\beta}},
\end{gather*}
and where $a_1$ is a constant and $\phi_{\beta}=\phi_{\beta}\big({{\bf x}}^{\beta}\big)$ are arbitrary smooth functions. Then the Helmholtz equation~\eqref{confHelmholtz} for the conformally rescaled metric $c^{4}g$ is $R$-separable in the $r$ groups of variables $\big({{\bf x}^{1}},\dots,{{\bf x}^{r}}\big)$. More precisely, if~$u$ is given by
\[
u=c^{-n+2} R w,
\]
with $R$ defined by
\begin{gather}\label{Rdef}
R=\frac{\big(s^{11}\big)^\frac{l_1}{4} \cdots \big(s^{r1}\big)^{\frac{l_r}{4}}}{(\det S)^{\frac{n-2}{4}}}.
\end{gather}
Then $w$ satisfies
\begin{gather}\label{sepeqw}
\sum_{\beta=1}^{r}\left(\frac{s^{\beta 1}}{\det S}\right)\big[{-}\Delta_{G_{\beta}}+\phi_{\beta}\big]w = a_1 w,
\end{gather}
which is separable in the $r$ groups of variables $\big({{\bf x}^{1}},\dots,{{\bf x}^{r}}\big)$ in the sense of \eqref{SeparabilitySchro}--\eqref{RankPainSchro} with the operators $B_\beta$ replaced by the operators $A_\beta = -\Delta_{G_{\beta}} + \phi_\beta$.
\end{Theorem}

We remark that the PDE of Yamabe type given by~(\ref{Yamabe1}) satisfied by the conformal fac\-tors~$c({\bf x})$ can be viewed as an extension of the generalized Robertson conditions to the setting of metrics that are \emph{conformally} Painlev\'e. Moreover, the existence of such conformal factors will be addressed in Section \ref{2} through Proposition~\ref{Yamabe}. In particular, it will be shown there that such metrics enlarge considerably the class of Painlev\'e metrics satisfying the generalized Robertson conditions~(\ref{GenRobertson}).

We conclude this section by referring to the interesting recent paper by Chanu and Rastelli~\cite{CR2019} that was published during the elaboration of the present paper. It turns out that Chanu and Rastelli define the Painlev\'e form of metrics like our Definition~\ref{Painleve} in connection with the notion of separability of the Hamilton--Jacobi equations in groups of variables. They provide several intrinsic characterizations of Painlev\'e metrics extending the ones stated in our Section~\ref{PainHelmSection}. We refer for instance to the beautiful invariant characterization of Painlev\'e metrics given in their Proposition~5.8 that allow them to classify all Painlev\'e metrics in~$\mathbb{E}_3$.

\section[Examples of Painlev\'e metrics satisfying the generalized Robertson conditions]{Examples of Painlev\'e metrics satisfying\\ the generalized Robertson conditions} \label{2}

In this section, we provide several examples of Painlev\'e metrics satisfying the generalized Robertson conditions (\ref{GenRobertson}) in all dimensions. Then we try to give a catalogue~-- as complete as possible~-- of such Painlev\'e metrics in dimensions $2$, $3$ and $4$. All our examples are local in the sense that they are defined in a single coordinate chart. Obtaining global examples of Riemannian or semi-Riemannian manifolds admitting an atlas of coordinate charts in which the metric is in Painlev\'e form appears to be a challenging task, well worthy of further investigation. This point will be discussed as one of the perspectives listed in Section~\ref{Perspectives}. From the notations used in definition (\ref{Painleve}), recall that a Painlev\'e metric is given in local coordinates $\big(x^1,\dots,x^n\big) = \big({\bf x}^{1},\dots,{\bf x}^{r}\big)$ where ${\bf x}^{\alpha}$ denotes group of variables indexed by $1\leq \alpha\leq r$ by
\begin{gather*} 
g = g_{ij}{\rm d}x^{i}{\rm d}x^{j}=\frac{\det S}{s^{11}} G_1 +\cdots+\frac{\det S}{s^{r1}} G_r,
\end{gather*}
for quadratic differential forms
\begin{gather} \label{Ga}
G_{\alpha}=G_{\alpha}\big({{\bf x}^{\alpha}}\big)=\sum_{i_{\alpha}=1_{\alpha}}^{l_{\alpha}}\sum_{j_{\alpha}=1_{\alpha}}^{l_{\alpha}}G_{\alpha}\big({{\bf x}^{\alpha}}\big)_{i_{\alpha}j_{\alpha}}{\rm d}x^{i_{\alpha}}{\rm d}x^{j_{\alpha}},\qquad 1\leq \alpha \leq r,
\end{gather}
and a generalized St\"ackel matrix $S$ of the form
\[
S=\left(
\begin{matrix}
s_{11}\big({{\bf x}^{1}}\big)&\dots&s_{1r}\big({{\bf x}^{1}}\big) \\
\vdots& & \vdots \\
s_{r1}\big({{\bf x}^{r}}\big)&\dots&s_{rr}\big({\bf x}^{r}\big)
\end{matrix}
\right).
\]

From (\ref{GenRobertson}), recall also that the Robertson conditions read
\[
\partial_{j_{\alpha}}\gamma_{i_{\beta}}=0,\qquad 1\leq \alpha \neq \beta \leq r, \quad 1_{\alpha}\leq i_{\alpha} \leq l_{\alpha},\quad 1_{\beta}\leq i_{\beta} \leq l_{\beta},
\]
where
\[
\gamma_{i_{\beta}}=-\partial_{i_{\beta}}\left[\log \frac{(\det S)^{\frac{n}{2}-1}s^{\beta 1}}{\big(s^{11}\big)^\frac{l_1}{2} \cdots \big(s^{r1}\big)^{\frac{l_r}{2}}}\right].
\]
Since $s^{\beta 1}$ does not depend on the group of variables $\bf{x^\beta}$, these conditions can be equivalently formulated as the algebraic-differential condition
\begin{gather} \label{RobertsonCond}
 \frac{(\det S)^{n-2}}{\big(s^{11}\big)^{l_1} \cdots \big(s^{r1}\big)^{l_r}} = \prod_{\alpha = 1}^r f_\alpha\big(\bf{x^\alpha}\big),
\end{gather}
where $f_\alpha = f_\alpha\big(\bf{x^\alpha}\big)$ are arbitrary functions of the indicated group of variables. We will use this last expression of the Robertson conditions to find different examples of Painlev\'e metrics in all dimensions that satisfy them. Our main examples are:

\begin{Example} \label{Exem0}
If $r = 2$ and $n=2$, then any St\"ackel matrix
\[
S = \left( \begin{matrix}
s_{11}\big( x^1\big) & s_{12}\big(x^1\big) \\
s_{21}\big(x^2\big)&s_{22}\big(x^2\big)
\end{matrix}
\right),
\]
satisfies automatically the usual Robertson conditions~(\ref{RobertsonCond}). The corresponding St\"ackel metrics in 2D can be given the following \emph{normal form}
\begin{gather*} 
 g = \big( f_1\big(x^1\big) + f_2\big(x^2\big) \big) \big( \big({\rm d}x^1\big)^2 + \big({\rm d}x^2\big)^2 \big),
\end{gather*}
where $f_\alpha$, $\alpha=1,2$ are arbitrary functions of ${\bf x}^\alpha$ such that $f_1 + f_2 > 0$. Thus we recover the classical Liouville metrics.

If $r = 2$ and $n \geq 3$, then any generalized St\"ackel matrix
\[
S = \left( \begin{matrix}
s_{11}\big({\bf x}^1\big) & s_{12}\big({\bf x}^1\big) \\
0&s_{22}\big({\bf x}^2\big)
\end{matrix}
\right),
\]
satisfy the generalized Robertson conditions~(\ref{RobertsonCond}). The corresponding Painlev\'e metrics can be given the following \emph{normal form}
\begin{gather*} 
 g = G_1 + f_1\big(\textbf{x}^1\big) G_2,
\end{gather*}
where $G_1$, $G_2$ are Riemannian metrics as in (\ref{Ga}) and $f_1 = f_1\big(\textbf{x}^1\big)$ is any positive function. Note that the metrics are classical \emph{warped products}.
\end{Example}

\begin{Example} \label{Exem1}Consider a generalized St\"ackel matrix of the form
\[
S=\left(
\begin{matrix}
s_{11}\big({{\bf x}^{1}}\big)&\dots&s_{1r}\big({{\bf x}^{1}}\big) \\
a_{21}&\dots&a_{2r} \\
\vdots& & \vdots \\
a_{r1}&\dots&a_{rr}
\end{matrix}
\right),
\]
where the entries $a_{\alpha \beta}$, $2\leq \alpha \leq r$, $1\leq \beta \leq r$, are real constants chosen such that~(\ref{StackelPositiveness}) is satisfied. Then it is immediate that
\[
 \frac{(\det S)^{n-2}}{\big(s^{11}\big)^{l_1} \cdots \big(s^{r1}\big)^{l_r}} = f_1\big({\bf x}^1\big),
\]
for some function $f_1$ depending only on ${\bf x}^1$. Hence the Robertson conditions~(\ref{RobertsonCond}) are trivially satisfied. The corresponding Painlev\'e metrics are of the general form of \emph{multiply warped products}
\begin{gather} \label{Painleve1}
 g = \sum_{\alpha =1}^r f_\alpha\big({\bf x}^1\big) G_\alpha,
\end{gather}
where $f_\alpha$ are arbitrary positive functions of ${\bf x}^1$ and $G_\alpha$ are given by~(\ref{Ga}). We note that the inverse anisotropic Calder\'on problem on a class of singular metrics of the form~(\ref{Painleve1}) is studied in~\cite{DKN2018a}.
\end{Example}

\begin{Example} \label{Exem2}
Our final class of examples is the most interesting one and comes from the theory of geodesically (or projectively) equivalent metrics (see for instance \cite{BoMa2015, Ma2005, Top1999, Top2002, TopMa2003}) and its link to particular St\"ackel systems called Benenti systems \cite{Ben2016, BoMa2003}. Note that it only applies to St\"ackel metrics satisfying the usual Robertson conditions, i.e., we assume that $r = n$ in the following. Note also that in the context of geodesically equivalent metrics, the commutation relations of Theorem~\ref{maintheorem2} were already established by direct computation in~\cite{MaTop2001}, and can also be seen to follow from the commutation of the corresponding Killing tensors with the Ricci tensor~\cite{MaKio2009}. Consider a St\"ackel matrix~$S$ of Vandermonde type (see \cite[Theorem~8.5]{Ben2016})
\[
 S = \big( (-1)^{n+\alpha - \beta + 1} f_\alpha^{n-\beta} \big)_{1 \leq \alpha, \beta \leq n},
 \]
where the functions $f_\alpha = f_\alpha(x^\alpha)$ only depend on the variable $x^\alpha$ and satisfy
\[
 f_1\big(x^1\big) < f_2\big(x^2\big) < \dots < f_r\big(x^r\big), \qquad \forall\, x^\alpha, \quad 1 \leq \alpha \leq n.
\]
An easy calculation shows that
\[
 \det(S) = \prod_{1 \leq \alpha < \beta \leq n} |f_\alpha - f_\beta|, \qquad s^{\gamma 1} = \prod_{\substack{ 1 \leq \alpha < \beta \leq n \\ \alpha, \beta \ne \gamma }} |f_\alpha - f_\beta|, \qquad \forall\, 1 \leq \gamma \leq n,
\]
from which we deduce that the Robertson conditions (\ref{RobertsonCond}) are satisfied. The corresponding St\"ackel metrics are given by
\begin{gather*} 
 g = \left( \prod_{\alpha \ne 1} |f_\alpha - f_1| \right) \big({\rm d}x^1\big)^2 + \left( \prod_{\alpha \ne 2} |f_\alpha - f_2| \right) \big({\rm d}x^2\big)^2 + \dots + \left( \prod_{\alpha \ne n} |f_\alpha - f_n| \right) \big({\rm d}x^n\big)^2.
\end{gather*}
\end{Example}

Let us now use the above classes of examples to give as exhaustively as possible a list of Painlev\'e metrics satisfying the generalized Robertson conditions in dimensions $n=2, 3, 4$. Note that the list of Painlev\'e metrics given below is only exhaustive as far as generic cases are concerned, and thus does not cover all the examples of Painlev\'e metrics satisfying the Robertson conditions, such as metrics of constant sectional curvature. Recall also that we always assume that $2 \leq r \leq n$.

\textbf{2D Painlev\'e metrics.} Let $n = 2$ and $r=2$. Then according to Example~\ref{Exem0}, the only Painlev\'e metrics are St\"ackel metrics given by
\begin{gather*} 
 g = \big( f_1\big(x^1\big) + f_2\big(x^2\big) \big) \big( \big({\rm d}x^1\big)^2 + \big({\rm d}x^2\big)^2 \big),
\end{gather*}
for some functions $f_1$ and $f_2$ such that $f_1 + f_2 > 0$. Hence Painlev\'e metrics satisfying the Robertson conditions in 2D are Liouville metrics.

\textbf{3D Painlev\'e metrics.} Let $n = 3$.
\begin{itemize}\itemsep=0pt
\item If $r=2$ and say $l_1= 1$, $l_2 = 2$, then according to Example~\ref{Exem0}, Painlev\'e metrics satisfying the generalized Robertson conditions are classical warped products; more precisely
\begin{gather*} 
 g = \big({\rm d}x^1\big)^2 + f_1\big(x^1\big) G_{2}, \qquad \textrm{or} \qquad g = f_2\big(\textbf{x}^2\big) \big({\rm d}x^1\big)^2 + G_2,
\end{gather*}
for some positive functions $f_1$ and $f_{2}$ depending only on the indicated groups of variables and any Riemannian metric
\[
 G_{2} = (G_{2})_{ij}\big(x^2,x^3\big) {\rm d}x^i {\rm d}x^j, \qquad i,j=2,3.
\]	

\item If $r=3$, then 3D Painlev\'e metrics are in fact St\"ackel metrics. According to Examp\-les~\ref{Exem1} and~\ref{Exem2}, we have the following possible expressions for St\"ackel metrics $g$ satisfying the Robertson conditions (see also~\cite{Gob2016}):
\begin{gather*} 
 g = f_1 \big({\rm d}x^1\big)^2 + h_1 \big({\rm d}x^2\big)^2 + k_1 \big({\rm d}x^3\big)^2,
\end{gather*}
or
\begin{gather*} 
 g = (f_3 - f_1)(f_2-f_1) \big({\rm d}x^1\big)^2 + (f_3-f_2)(f_2-f_1) \big({\rm d}x^2\big)^2 + (f_3-f_1)(f_3-f_2) \big({\rm d}x^3\big)^2,
\end{gather*}
where $f_1$, $h_1$, $k_1$ are functions of the variable $x^1$ only and $f_2$, $f_3$ are functions of the variables $x^2$ and $x^3$ respectively such that $f_1 < f_2 < f_3$. We add a last example to this list found by inspection of the Robertson conditions~(\ref{RobertsonCond}). Consider the St\"ackel matrix
\[
S=\left( \begin{matrix}
1&s_{12}&a s_{13}\\
0&s_{22}&s_{23} \\
0&s_{32}&s_{33}
\end{matrix} \right),
\]
where $a$ is a real constant and the $s_{ij}=s_{ij}\big(x^i\big)$ are arbitrary functions of the indicated variables for which $\det S \neq 0$. Then we can check directly that the Robertson conditions are satisfied and we obtain the following expression for the corresponding St\"ackel metrics
\begin{gather} \label{Painleve3D-4}
 g = \big({\rm d}x^1\big)^2 + \left( \frac{1}{s_{12}} \right) \left[ (s_{22}s_{33} - s_{23}s_{32}) \left( \frac{\big({\rm d}x^2\big)^2}{s_{32}-as_{33}} + \frac{\big({\rm d}x^3\big)^2}{s_{23}-as_{22}} \right) \right].
\end{gather}
Note in particular that such metrics are warped products and thus admit a conformal Killing vector field.
\end{itemize}

\textbf{4D Painlev\'e metrics.} Let $n=4$.
\begin{itemize}\itemsep=0pt
\item If $r=2$ and $l_1 + l_2 = 4$, then according to Example~\ref{Exem0}, Painlev\'e metrics that satisfy the generalized Robertson conditions are warped products of the type
\begin{gather*} 
 g = G_1 + f_1\big(\textbf{x}^1\big) G_{2}, \qquad \textrm{or} \qquad g = f_2\big(\textbf{x}^2\big) G_1 + G_{2},
\end{gather*}
for some positive functions $f_1$ and $f_{2}$ depending only on the indicated variables and any Riemannian metrics $G_1$, $G_2$ of the type~(\ref{Ga}).
\item If $r=3$ and $l_1 = 2$, $l_2 = l_3 = 1$ (the other cases are treated similarly), then according to Example~\ref{Exem1}, we obtain the following Painlev\'e metrics
\begin{gather*} 
 g = h G_1 + k \big({\rm d}x^3\big)^2 + l \big({\rm d}x^4\big)^2,
\end{gather*}
where $h$, $k$, $l$ are positive functions of the variables $x^1$, $x^2$ only and
\begin{gather*} G_1 = (G_1)_{ij}\big(x^1,x^2\big) {\rm d}x^i {\rm d}x^j, \qquad i,j=1,2\end{gather*} is any Riemannian metric. Following the same procedure as in example (\ref{Painleve3D-4}), we also obtain the following class of Painlev\'e metrics
\begin{gather*} 
 g = G_1 + \left( \frac{1}{s_{12}} \right) \left[ (s_{22}s_{33} - s_{23}s_{32}) \left( \frac{\big({\rm d}x^3\big)^2}{s_{32}-as_{33}} + \frac{\big({\rm d}x^4\big)^2}{s_{23}-as_{22}} \right) \right],
\end{gather*}
where $s_{12} = s_{12}\big(x^1,x^2\big)$, $s_{22} = s_{22}\big(x^3\big)$, $s_{23} = s_{23}\big(x^3\big)$, $s_{32} = s_{32}\big(x^4\big)$, $s_{33} = s_{33}\big(x^4\big)$ and $G_1 = (G_1)_{ij}\big(x^1,x^2\big) {\rm d}x^i {\rm d}x^j$, $i,j=1,2$. Note in particular that such metrics are warped products that admit a~conformal Killing vector field.

\item If $r=4$, the Painlev\'e metrics are St\"ackel metrics. According to Examples~\ref{Exem1} and~\ref{Exem2}, possible expressions for Painlev\'e metrics satisfying the Robertson conditions are
\begin{gather*} 
 g = f_1 \big({\rm d}x^1\big)^2 + h_1 \big({\rm d}x^2\big)^2 + k_1 \big({\rm d}x^3\big)^2 + l_1 \big({\rm d}x^4\big)^2,
\end{gather*}
or
\begin{gather*}
 g = (f_4 - f_1)(f_3-f_1)(f_2-f_1) \big({\rm d}x^1\big)^2 + (f_4-f_2)(f_3-f_2)(f_2-f_1) \big({\rm d}x^2\big)^2\nonumber \\
 \hphantom{g =}{} + (f_4-f_3)(f_3-f_1)(f_3-f_2) \big({\rm d}x^3\big)^2 + (f_4-f_1)(f_4-f_2)(f_4-f_3) \big({\rm d}x^4\big)^2, 
\end{gather*}
where $f_1$, $h_1$, $k_1$, $l_1$ are functions of the variable $x^1$ only and $f_2$, $f_3$, $f_4$ are functions of the variables $x^2$, $x^3$ and $x^4$ respectively such that $f_1 < f_2 < f_3 < f_4$. We add a last example to this list found by inspection of the Robertson conditions~(\ref{RobertsonCond}). Consider the St\"ackel matrix
\[
S=\left( \begin{matrix}
1&s_{12}&a s_{12}&s_{12}\\
0&1&s_{23}&s_{23} \\
0&0&s_{33}&s_{34} \\
0&0&s_{43}&s_{44}
\end{matrix} \right),
\]
where $s_{ij}=s_{ij}\big(x^i\big)$ arbitrary functions of the indicated variables. Then we can check directly that the generalized Robertson conditions are satisfied and we obtain the following expression for the corresponding St\"ackel metrics
\begin{gather*} 
 g = \big({\rm d}x^1\big)^2 + \left( \frac{1}{s_{12}} \right) \Bigg[ \big({\rm d}x^2\big)^2 \\
 \hphantom{g =}{} + \left( \frac{1}{s_{23}-1}\right) \left( (s_{33}s_{44} - s_{34}s_{43}) \left( \frac{\big({\rm d}x^3\big)^2}{s_{44}-s_{43}} + \frac{\big({\rm d}x^4\big)^2}{s_{33}-s_{34}} \right) \right) \Bigg].
\end{gather*}
Note that such metrics are warped products and that the metrics between square brackets are also warped products.
\end{itemize}

We end this section by giving some existence results for the conformal factor $c(\bf x)$ appearing in Theorem~\ref{maintheorem3}, in the case in which $M$ is a smooth compact manifold of dimension $n \geq 3$, with smooth boundary $\partial M$. We recall from Theorem~\ref{maintheorem3} that the conformal factor $c(\bf x)$ must satisfy a non-linear PDE of Yamabe type, given by
\begin{gather*}
\Delta_{g}c^{n-2}+f({\bf x}) c^{n-2} - \lambda c^{n+2}=0,
\end{gather*}
where
\begin{gather*}
f({\bf x})= \left( \sum_{\beta =1}^{r}\frac{s^{\beta 1}}{\det S}(\phi_{\beta}-P_{\beta})\right) - a_1 ,
\end{gather*}
and where $\phi_{\beta}=\phi_{\beta}\big({{\bf x}}^{\beta}\big)$ are arbitrary smooth functions. Setting $w = c^{n-2}$, we are thus interested in solutions $w = c^{n-2}$ of the non-linear elliptic PDE:
\begin{alignat}{3}
& \Delta_g w + f({\bf x})w - \lambda w^{\frac{n+2}{n-2}} =0 , \qquad && \textrm{on} \ M, &\nonumber\\
&  w = \eta, && \textrm{on} \ \partial M, & \label{Eqw}
\end{alignat}
where $\eta$ is any suitable smooth positive function on $\partial M$. We can solve~(\ref{Eqw}) by using the well-known technique of lower and upper solutions which we briefly recall here. Setting
\[
f({\bf x}, w) = f({\bf x})w - \lambda w^{\frac{n+2}{n-2}},
\]
we recall that an upper solution ${\overline{w}}$ is a function in $C^2(M) \cap C^0(\overline{M})$ satisfying
\begin{gather*}
\Delta_g {\overline{w}}+ f({\bf x},{\overline{w}}) \leq 0 \quad \textrm{on} \ M, \qquad \textrm{and} \qquad {\overline{w}}_{|\partial M} \geq \eta.
\end{gather*}
Similarly, a lower solution ${\underline{w}}$ is a function in $C^2(M) \cap C^0(\overline{M})$ satisfying
\begin{gather*}
\Delta_g {\underline{w}}+ f({\bf x},{\underline{w}}) \geq 0 \quad \textrm{on} \ M, \qquad \textrm{and} \qquad {\underline{w}}_{|\partial M} \leq \eta.
\end{gather*}

It is well-known that if we can find a lower solution ${\underline{w}}$ and an upper solution ${\overline{w}}$ satisfying ${\underline{w}} \leq {\overline{w}}$ on~$M$, then there exists a solution $w \in C^{\infty}(\overline{M})$ of~(\ref{Eqw}) such that ${\underline{w}} \leq w \leq {\overline{w}}$ on $M$.

Now, we can prove the following result:

\begin{Proposition} \label{Yamabe}\quad
\begin{enumerate}\itemsep=0pt
\item[$1.$] If $\lambda >0$ and $f({\bf x})>0$ on $M$, then for each positive function $\eta$ on $\partial M$, there exists a~smooth positive solution $w$ of~\eqref{Eqw}.
\item[$2.$] If $\lambda \leq 0$ and $f({\bf x}) < \lambda$ on $M$, then for each for each positive function $\eta$ on $\partial M$ such that $\eta \leq 1$, there exists a smooth positive solution $w$ of~\eqref{Eqw}.
    \end{enumerate}	
\end{Proposition}

\begin{Remark}Since $\frac{s^{\beta 1}}{\det S} >0$ by the hypothesis~(\ref{StackelPositiveness}), we see that the assumption $f({\bf x})>0$ on~$M$ (resp.~$f({\bf x}) < \lambda$ on~$M$) is satisfied if the $\phi_{\beta}$'s are chosen sufficiently large (resp.~$-\phi_{\beta}$ are sufficiently large).
\end{Remark}

\begin{proof}1. We use the technique of lower and upper solutions. We define ${\underline{w}}= \epsilon$ where $\epsilon>0$ is small enough. Thus, $\underline{w} \leq \eta$ on
$\partial M$ and we have
\[
\Delta_g \underline{w} + f({\bf x}, \underline{w} ) =
\epsilon \big( f({\bf x}) -\lambda \epsilon^{{\frac{n+2}{n-2}} -1} \big) >0,
\]
so $\underline{w}$ is a lower solution. In the same way, we define $\overline{w} = C$ where $C$ is sufficiently large. Thus $\overline{w} \geq \eta$ and we have
\[
\Delta_g \overline{w} + f({\bf x}, \overline{w}) = C \big( f({\bf x}) -\lambda C^{{\frac{n+2}{n-2}} -1} \big) <0.
\]
It follows that $\overline{w}$ is an upper solution and clearly $\underline{w} \leq \overline{w}$. Thus, there exist a smooth positive solution~$w$ of~(\ref{Eqw}) satisfying $\epsilon \leq w \leq C$.

2. In the case $\lambda \leq 0$, $f({\bf x}) < \lambda$ on $M$ and $\eta \leq 1$ on $\partial M$, we define ${\underline{w}}$ as the unique solution of the Dirichlet problem
\begin{alignat*}{3}
 & \Delta_g \underline{w} + f({\bf x}) \underline{w} = 0 ,\qquad & & \textrm{on} \ M,&\nonumber \\
 & \underline{w} = \eta, \qquad && \textrm{on} \ \partial M.& 
\end{alignat*}
The strong maximum principle implies that $0 < \underline{w} \leq \max \eta$ on $M$.
Moreover, $\triangle_g \underline{w} + f({\bf x}, \underline{w}) = -\lambda (\underline{w})^{\frac{n+2}{n-2}} \geq 0$. Hence $\underline{w}$ is a lower solution of~(\ref{Eqw}).

Now, we define ${\overline{w}}$ as the unique solution of the Dirichlet problem
\begin{alignat*}{3} 
 & \Delta_g \overline{w} + f({\bf x}) \overline{w} = f({\bf x}) (\max \eta)^{\frac{n+2}{n-2}} , \qquad && \textrm{on} \ M,& \\
 & \overline{w} = \eta, && \textrm{on} \ \partial M.&
\end{alignat*}
According to the maximum principle, we also have $\overline{w} \geq 0$ on $M$. Setting $v = \overline{w}- \max \eta$, we see that
\[
\Delta_g v + f({\bf x}) v = f({\bf x}) \big(\max \eta^{\frac{n+2}{n-2}} - \max \eta\big) \geq 0,
\]
since $\eta \leq 1$. Hence, the maximum principle implies that $v \leq 0$ on~$M$, or equivalently $ \overline{w} \leq \max \eta$.

We deduce that
\[
\Delta_g \overline{w} + f({\bf x}, \overline{w}) =
f({\bf x})\big(\max \eta^{\frac{n+2}{n-2}} - \overline{w}^{\frac{n+2}{n-2}}\big) +(f({\bf x})-\lambda) \overline{w}^{\frac{n+2}{n-2}} \leq 0,
\]
since $f({\bf x})< \lambda$. Thus, $\overline{w}$ is an upper solution of (\ref{Eqw}). Finally, $\overline{w} - \underline{w}$ satisfies
\begin{alignat*}{3} 
&\Delta_g (\overline{w} - \underline{w}) +f({\bf x}) (\overline{w} - \underline{w}) = f({\bf x}) (\max \eta)^{\frac{n+2}{n-2}} < 0 , \qquad && \textrm{on} \ M,& \\
& \overline{w} - \underline{w} = 0, && \textrm{on} \ \partial M. &
\end{alignat*}
Then, the maximum principle implies again $\overline{w} \geq \underline{w}$. Then according to the lower and upper solutions technique, there exists a smooth positive solution $w$ of~(\ref{Eqw}).
\end{proof}

We conclude by remarking that there exist important classes of $n$-dimensional metrics of physical interest for which the geodesic flow admits $[n/2]-1$ Poisson-commuting quadratic first integrals arising from the presence of a principal Killing--Yano tensor with torsion. We refer to~\cite{Frolov2017, Houri2012} for results on their local classification and normal forms.

\section{Generalized Killing--Eisenhart and Levi-Civita conditions}

The proofs of the main results of our paper, that is Theorems \ref{maintheorem1}, \ref{maintheorem2} and \ref{maintheorem3}, make use of generalizations to Painlev\'e metrics of the classical Killing--Eisenhart equations and Levi-Civita separability conditions which hold for St\"ackel metrics (see for example \cite{BCR1-2002, BCR2-2002, Eis1934, KaMi1980, KaMi1981, Per1990}). We present these generalizations in the form of the following two lemmas, beginning with the Killing--Eisenhart equations. Thus in analogy with the St\"ackel case, we introduce the quantities
\begin{gather}\label{defrho}
\rho_{\beta \gamma}:=\frac{s^{\gamma \beta}}{s^{\gamma 1}}.
\end{gather}
Note that by the assumption (\ref{StackelPositiveness}), we have $s^{\gamma 1}\neq 0$.
The following lemma gives the generalization to the case of Painlev\'e metrics of the Killing--Eisenhart equations given in \cite{BCR1-2002,BCR2-2002,Eis1934,KaMi1980,KaMi1981} for St\"ackel metrics:
\begin{Lemma} \label{KE}
We have, for all $1\leq \beta, \delta, \gamma \leq r$, the identities
\begin{gather}\label{genKE}
\partial_{j_{\gamma}}(\rho_{ \beta \delta})= (\rho_{\beta \gamma}-\rho_{\beta \delta })\left(\partial_{j_{\gamma}}{\log \frac{s^{\delta 1}}{\det S}}\right).
\end{gather}
\end{Lemma}
We will refer to \eqref{genKE} as the {\emph{generalized Killing--Eisenhart equations.}}
\begin{proof}The Poisson bracket relations (\ref{PoissonCommute}) imply
\begin{gather}\label{Painlevecomm}
\sum_{\pgamma=1}^{\lgamma}\big(\big(\partial_{\pgamma} K_{(1)}^{\idelta \jdelta}\big) K_{(\beta)}^{\pgamma \kgamma}-\big(\partial_{\pgamma} K_{(\beta)}^{\idelta \jdelta}\big) K_{(1)}^{\pgamma \kgamma}\big)=0,
\end{gather}
where $1\leq \idelta, \jdelta \leq \ldelta$, $1\leq \kgamma \leq \lgamma$ and $1\leq \delta, \gamma \leq r$. Using the expressions~(\ref{Killingtensor}) of the Killing tensors $\big(K_{(\beta)}^{\idelta \jdelta}\big)$, the fact that $\partial_{\igamma}\big(G^{\beta}\big)^{\ibeta \jbeta}=0$ for $\beta \neq \gamma$, and the fact that each of the $\lgamma \times \lgamma$ matrices $\big(\big(G^{\gamma}\big)^{\igamma \jgamma}\big)$, is invertible, we obtain that the relations~(\ref{Painlevecomm}) are equivalent to
\begin{gather}\label{Stackelcomm}
\left[\partial_{\pgamma}\left(\frac{s^{\delta 1}}{\det S}\right)\right] \frac{s^{\gamma \beta}}{\det S}-\left[\partial_{\pgamma}\left(\frac{s^{\delta \beta}}{\det S}\right)\right] \frac{s^{\gamma 1}}{\det S}=0,
\end{gather}
where $1\leq \delta, \gamma \leq r$, $1_{\gamma}\leq \pgamma \leq \lgamma $. In particular, the relations (\ref{Painlevecomm}) are independent of the block metrics (\ref{blockmetric}). Setting $\delta = \gamma$ in~(\ref{Stackelcomm}), we obtain
\[
\partial_{\pgamma}\left(\frac{s^{\gamma \beta}}{s^{\gamma 1}}\right)=0,
\]
for $1_{\gamma}\leq \pgamma \leq \lgamma$, so that using the definition (\ref{defrho}) of the quantities $\rho_{\beta \gamma}$, the relations (\ref{Stackelcomm}) take the form
\[
\partial_{\pgamma}\left(\rho_{ \beta \delta}\frac{s^{\delta 1}}{\det S}\right)\frac{s^{\gamma 1}}{\det S}=\left(\partial_{\pgamma}\left(\frac{s^{\delta 1}}{\det S}\right)\right)\rho_{\beta \gamma} \frac{s^{\gamma 1}}{\det S},
\]
which in turn reduces to (\ref{genKE}), thus proving our claim.
\end{proof}

In order to state the generalization to Painlev\'e metrics of the classical Levi-Civita conditions which hold for St\"ackel metrics, we make the hypothesis
\begin{gather}\label{H}
\rho_{\beta \delta}\neq \rho_{\beta \epsilon},\qquad \forall\, 1\leq \beta \leq r,\quad \forall\, 1\leq \delta\neq \epsilon \leq r.
\end{gather}

The generalization of the Levi-Civita conditions to the Painlev\'e case is now given by the following:
\begin{Lemma} \label{LeviCivita}
A generalized St\"ackel matrix \eqref{genStackelmat} corresponding to a Painlev\'e metric~\eqref{Painleveform} for which the genericity hypothesis~\eqref{H} holds true satisfies, for all $1\leq \beta, \gamma \leq r$,
the identities
\begin{gather}
\left(\partial_{\jalpha}\log\left(\frac{s^{\gamma 1}}{\det S}\right)\right) \left(\partial_{\kbeta}\log\left(\frac{s^{\alpha 1}}{\det S}\right)\right)+\left(\partial_{\jalpha}\log\left(\frac{s^{\beta 1}}{\det S}\right)\right)\left(\partial_{\kbeta}\log\left(\frac{s^{\gamma 1}}{\det S}\right)\right)\nonumber\\
\qquad{} -\frac{\det S}{s^{\gamma 1}}\partial_{\jalpha}\partial_{\kbeta}\left(\frac{s^{\gamma 1}}{\det S}\right)=0.\label{genLeviCivita}
\end{gather}
In particular, we have the identities
\begin{gather}\label{secondorder}
\frac{\partial}{\partial x^{\jalpha}}\frac{\partial }{\partial x^{\kbeta}}\left(\frac{\det S}{s^{\alpha1}s^{\beta1}}\right)=0,
\end{gather}
for all
$1\leq \alpha, \beta \leq r$.
\end{Lemma}
We will likewise refer to the identities (\ref{genLeviCivita}) as the {\emph{generalized Levi-Civita conditions}}.
\begin{Remark}\quad
\begin{enumerate}\itemsep=0pt
\item[1)] In the case of St\"ackel metrics, that is when $r=n$, the conditions (\ref{genLeviCivita}) reduce to the classical Levi-Civita conditions, given by
\begin{gather*}
\left(\partial_{j}\log\left(\frac{s^{l1}}{\det S}\right)\right)\left(\partial_{k}\log\left(\frac{s^{j1}}{\det S}\right)\right)+\left(\partial_{j}\log\left(\frac{s^{k1}}{\det S}\right)\right)\left(\partial_{k}\log \left(\frac{s^{l1}}{\det S}\right)\right)\\
\qquad{} -\frac{\det S}{s^{l1}}\partial_{j}\partial_{k}\left(\frac{s^{l1}}{\det S}\right)=0.
\end{gather*}
\item[2)] We shall show at the end of the next Section~\ref{PainHJSection} that the generalized Levi-Civita condi\-tions~(\ref{genLeviCivita}) hold in fact for all Painlev\'e metrics~(\ref{Painleveform}) without assuming our genericity hypothesis~(\ref{H}). Nevertheless, it is easier to obtain them from the Killing--Eisenhart equations under the assumption~(\ref{H}) as we do below.
\end{enumerate}
\end{Remark}
\begin{proof}
The general idea behind the proof is similar to the one that is used in the classical St\"ackel case, and is based on expressing the integrability conditions for the generalized Killing--Eisenhart (\ref{genKE}), with additional twist resulting from the fact that the separation is in groups of variables only. We let $1\leq \alpha \neq \beta \leq r$ denote fixed indices and introduce the simplified notation
\[
\rho_{\delta}:=\rho_{ \beta \delta}=\frac{s^{\delta \beta}}{s^{\delta 1}},
\]
so as to make the expressions a bit more compact. The generalized Killing--Eisenhart equa\-tions~(\ref{genKE}) thus take the form
\begin{gather*}
\partial_{\pgamma} \rho_{\delta}= (\rho_{\gamma}-\rho_{\delta })\left(\partial_{\pgamma}\left(\log \frac{s^{\delta 1}}{\det S}\right)\right),
\end{gather*}
where $1\leq \delta,\gamma \leq r$. The integrability conditions
\begin{gather}\label{intKE}
\partial_{\kepsilon}\big(\partial_{\pgamma} \rhodelta \big)=\partial_{\pgamma}\big(\partial_{\kepsilon} \rhodelta \big)
\end{gather}
are now easily obtained. We have
\begin{gather*}
\partial_{\kepsilon} (\partial_{\pgamma} \rhodelta )= \left[(\rhoepsilon-\rhogamma)\partial_{\kepsilon}\log\left(\frac{s^{\gamma 1}}{\det S}\right)-(\rhoepsilon-\rhodelta)\partial_{\kepsilon}\log\left(\frac{s^{\delta 1}}{\det S}\right)\right] \partial_{\pgamma} \log \frac{s^{\delta 1}}{\det S} \\
\hphantom{\partial_{\kepsilon} (\partial_{\pgamma} \rhodelta )=}{} + (\rhogamma-\rhodelta)\partial_{\kepsilon}\partial_{\pgamma}\log\left(\frac{s^{\delta 1}}{\det S}\right),
\end{gather*}
so that (\ref{intKE}) becomes
\begin{gather}
(\rho_{\gamma}-\rho_{\epsilon})\left[\left(\partial_{\kepsilon}\log\left(\frac{s^{\gamma 1}}{\det S}\right)\right)\left(\partial_{\pgamma}\log\left(\frac{s^{\delta 1}}{\det S}\right)\right)\right.\nonumber\\
\qquad{} +\left(\partial_{\kepsilon}\log\left(\frac{s^{\delta 1}}{\det S}\right)\right)\left(\partial_{\pgamma}\log\left(\frac{s^{\epsilon 1}}{\det S}\right)\right) \nonumber\\
\left.\qquad{} -\left(\partial_{\kepsilon}\log\left(\frac{s^{\delta 1}}{\det S}\right)\right)\left(\partial_{\pgamma}\log\left(\frac{s^{\delta 1}}{\det S}\right)\right)-\partial_{\kepsilon}\partial_{\pgamma}\log\left(\frac{s^{\delta 1}}{\det S}\right)\right]=0,\label{compat}
\end{gather}
where $1\leq \epsilon \neq \gamma \leq r$ and $1\leq \delta \leq r$. Using the rank hypothesis (\ref{H}), expanding the logarithmic second derivative in (\ref{compat}) using the identity
\begin{gather}\label{identity}
\partial_{xy}\log f = \frac{1}{f}\partial_{x}\partial_{y}f-(\partial_{x}\log f)(\partial_{y}\log f)
\end{gather}
and
 relabeling the indices, we obtain
\begin{gather*}
\left(\partial_{\jalpha}\log\left(\frac{s^{\gamma 1}}{\det S}\right)\right) \left(\partial_{\kbeta}\log\left(\frac{s^{\alpha 1}}{\det S}\right)\right)+\left(\partial_{\jalpha}\log\left(\frac{s^{\beta 1}}{\det S}\right)\right)\left(\partial_{\kbeta}\log\left(\frac{s^{\gamma 1}}{\det S}\right)\right)\\
\qquad{}-\frac{\det S}{s^{\gamma 1}}\partial_{\jalpha}\partial_{\kbeta}\left(\frac{s^{\gamma 1}}{\det S}\right)=0,
\end{gather*}
which is precisely (\ref{genLeviCivita}). Finally, the relations~(\ref{secondorder}) are obtained by setting $\delta = \epsilon$ in the integrability condition (\ref{compat}), using the identity (\ref{identity}), and the fact that the cofactors $s^{\gamma \alpha}$ and $s^{\epsilon \alpha}$ are independent of the groups of variables ${\bf x}^{\gamma}$ and ${\bf x}^{\epsilon}$ respectively.
\end{proof}

\section{Different characterizations of Painlev\'e metrics}\label{PainHJSection}

Let us start with the characterization of Painlev\'e metrics in terms of complete additive separability of the Hamilton--Jacobi equations stated in Proposition \ref{HJSeparability-Painleve}.

\begin{proof}[Proof of Proposition \ref{HJSeparability-Painleve}]

In block orthogonal coordinates $(\textbf{x}^\alpha)$ associated to the metric
\[
 g = \sum_{\alpha = 1}^r c_\alpha G_\alpha,
\]
the Hamilton--Jacobi equation (\ref{HJ}) reads
\begin{gather} \label{HJ1}
 \sum_{\alpha =1}^{r} c^\alpha \sum_{\ialpha =1_{\alpha}}^{\lalpha}\sum_{\jalpha =1_{\alpha}}^{\lalpha} \big(G^{\alpha}\big)^{\ialpha \jalpha}\frac{\partial W}{\partial x^{\ialpha} }\frac{\partial W}{\partial x^{\jalpha}}=E=:a_{1},
\end{gather}
where $c^\alpha = (c_\alpha)^{-1}$. Assume that there exists a solution $W$ in the block-separable form~(\ref{Painsepform}) satisfying the completeness condition~(\ref{rankPainHJ}). Then (\ref{HJ1}) can be written as
\begin{gather} \label{HJ2}
 \sum_{\alpha =1}^{r} c^\alpha \sum_{\ialpha =1_{\alpha}}^{\lalpha}\sum_{\jalpha =1_{\alpha}}^{\lalpha} \big(G^{\alpha}\big)^{\ialpha \jalpha}\frac{\partial W_\alpha}{\partial x^{\ialpha} }\frac{\partial W_\alpha}{\partial x^{\jalpha}}=E=:a_{1}.
\end{gather}
Differentiating (\ref{HJ2}) with respect to $a_\beta$, we get
\begin{gather} \label{HJ3}
 \sum_{\alpha =1}^{r} c^\alpha \sum_{\ialpha =1_{\alpha}}^{\lalpha}\sum_{\jalpha =1_{\alpha}}^{\lalpha} 2 \big(G^{\alpha}\big)^{\ialpha \jalpha}\frac{\partial W_\alpha}{\partial x^{\ialpha} } \frac{\partial^2 W_\alpha}{\partial a_\beta \partial x^{\jalpha}} = \delta_{1\beta} .
\end{gather}
From (\ref{Painsepform}) and (\ref{rankPainHJ}), it follows that the family of matrices
\begin{gather*} 
 S(a_1,\dots,a_r) = (S_{\alpha \beta})(a_1,\dots,a_r) := 2 \big(G^{\alpha}\big)^{\ialpha \jalpha}\frac{\partial W_\alpha}{\partial x^{\ialpha} } \frac{\partial^2 W_\alpha}{\partial a_\beta \partial x^{\jalpha}}
\end{gather*}
are non-singular st\"ackel matrices of rank $r$ for all $(a_1,\dots,a_r) \in U$. Using the invertibility of $S(a_1,\dots,a_r)$ which is equivalent to the completeness condition (\ref{rankPainHJ}), we get immediately from~(\ref{HJ3}) that
\[
 c^\alpha = \frac{s^{\alpha 1}}{\det S},
\]
which proves that $g$ is a Painlev\'e metric.

Conversely, assume that $g$ is a Painlev\'e metric of the form (\ref{Painleveform}). Then the Hamilton--Jacobi equation (\ref{HJ}) takes the form
\begin{gather}\label{HJPainform}
\sum_{\alpha =1}^{r} \left(\frac{s^{\alpha 1}}{\det S}\right) \sum_{\ialpha =1_{\alpha}}^{\lalpha}\sum_{\jalpha =1_{\alpha}}^{\lalpha} \big(G^{\alpha}\big)^{\ialpha \jalpha}\frac{\partial W}{\partial x^{\ialpha} }\frac{\partial W}{\partial x^{\jalpha}}=E=:a_{1}.
\end{gather}
Now, since
\[
\sum_{\alpha =1}^{r}\left(\frac{s^{\alpha 1}}{\det S}\right)s_{\alpha \beta}=\delta_{1 \beta},
\]
we may rewrite (\ref{HJPainform}) as
\begin{gather} \label{HJ4}
\sum_{\alpha =1}^{r}\frac{s^{\alpha 1}}{\det S}\left(\sum_{\ialpha =1_{\alpha}}^{\lalpha}\sum_{\jalpha =1_{\alpha}}^{\lalpha}\big(G^{\alpha}\big)^{\ialpha \jalpha}\frac{\partial W}{\partial x^{\ialpha} }\frac{\partial W}{\partial x^{\jalpha}}-\sum_{\beta=1}^{r}s_{\alpha \beta}a_{\beta}\right)=0,
\end{gather}
for any $(a_1,\dots,a_r) \in U \subset \mathbb{R}^r$. Choosing $W$ in the block-separable form~(\ref{Painsepform}), we see that any solution of the reduced Hamilton--Jacobi equations
\begin{gather} \label{HJ5}
\sum_{\ialpha =1_{\alpha}}^{\lalpha}\sum_{\jalpha =1_{\alpha}}^{\lalpha}\big(G^{\alpha}\big)^{\ialpha \jalpha}\frac{\partial W_{\alpha}}{\partial x^{\ialpha} }\frac{\partial W_{\alpha}}{\partial x^{\jalpha}}=\sum_{\beta = 1}^{r}s_{\alpha \beta}a_{\beta},
\end{gather}
will provide a solution of (\ref{HJ4}). But the latter equation always admit \emph{locally} solutions $W_\alpha(\textbf{x}^\alpha, \allowbreak a_1,\dots,a_r)$ by standard PDE results \cite{Tay2011a}. Differentiating (\ref{HJ5}) with respect to $a_{\beta}$, we obtain
\begin{gather*} 
2\sum_{\alpha =1}^{r}\sum_{\ialpha =1_{\alpha}}^{\lalpha}\sum_{\jalpha =1_{\alpha}}^{\lalpha}\big(G^{\alpha}\big)^{\ialpha \jalpha}\left(\frac{\partial W_{\alpha}}{\partial x^{\ialpha}}\right)\left( \frac{\partial^{2}W_{\alpha}}{\partial a_{\beta}\partial x^{\jalpha}}\right)=s_{\alpha \gamma}.
\end{gather*}
Since the generalized St\"ackel matrix (\ref{genStackelmat}) is to be non-singular, it follows that the rank condition that must be satisfied by our block-separable parametrized family of solutions of the Hamilton--Jacobi equation is precisely~(\ref{rankPainHJ}), thus proving that metrics of the Painlev\'e form~(\ref{Painleveform}) are indeed characterized by the existence of a parametrized family of solutions of the Hamilton--Jacobi equation satisfying a suitable completeness condition.
\end{proof}

Let us now give another proof of Proposition~\ref{HJSeparability-Painleve} that will allow us to characterize Painlev\'e metrics in terms of the generalized Levi-Civita conditions (\ref{genLeviCivita}). Working in block-orthogonal coordinates or directly with a Painlev\'e metric with the identification
\begin{gather}\label{not}
c^\alpha:=\frac{\det S}{s^{\alpha 1}},
\end{gather}
the Hamilton--Jacobi equation (\ref{HJ}) takes the form
\begin{gather}\label{HJPainformnew}
\sum_{\alpha =1}^{r} c^\alpha \sum_{\ialpha =1_{\alpha}}^{\lalpha}\sum_{\jalpha =1_{\alpha}}^{\lalpha} \big(G^{\alpha}\big)^{\ialpha \jalpha}\frac{\partial W}{\partial x^{\ialpha} }\frac{\partial W}{\partial x^{\jalpha}}=E=:a_{1}.
\end{gather}
We recall that we seek a solution $W$ of (\ref{HJPainformnew}) which is additively separable into groups of variables, that is
\[
W=\sum_{\alpha =1}^{r}W_{\alpha}\big({\bf x}^{\alpha}\big),
\]
and let
\[
u_{\alpha}^{(1)}=\sum_{\ialpha =1_{\alpha}}^{\lalpha}\sum_{\jalpha =1_{\alpha}}^{\lalpha}\big(G^{\alpha}\big)^{\ialpha \jalpha}\frac{\partial W_{\alpha}}{\partial x^{\ialpha} }\frac{\partial W_{\alpha}}{\partial x^{\jalpha}},
\]
in which case the Hamilton--Jacobi equation (\ref{HJPainformnew}) takes the form
\[
\sum_{\alpha =1}^{r}c^{\alpha}u_{\alpha}^{(1)}=a_{1}.
\]
Differentiating the latter equation with respect to $x^{\kbeta}$, we obtain the first order differential system in normal form
\begin{gather}\label{normalsyst}
 \partial_{\kbeta}u^{(1)}_{\gamma} =   \begin{cases}\displaystyle -\frac{1}{c^{\beta}}\sum_{\alpha =1}^{r}\big(\partial_{\kbeta} c^{\alpha}\big) u^{(1)}_{\alpha} & \gamma = \beta, \\
 0 & \gamma \neq \beta . \end{cases}
\end{gather}
Introducing the first-order differential operators
\begin{gather}\label{totaldiff}
D_{\kbeta} := \frac{\partial}{\partial x^{\kbeta}}-\frac{1}{c^{\beta}}\sum_{\alpha =1}^{r}\big(\partial_{\kbeta} c^{\alpha}\big) u^{(1)}_{\alpha}\frac{\partial}{\partial u^{(1)}_{\beta}},
\end{gather}
the differential system (\ref{normalsyst}) will admit a family of solutions $u^{(1)}_{\alpha}=u^{(1)}_{\alpha}\big(x^{1},\dots,x^{n};a_{1},\dots,a_{r}\big)$ defined on an open subset $U\in M$ and depending smoothly on~$r$ constants $(a_{1},\dots,a_{r})$ defined in an open subset $A\subset {\mathbb{R}}^{r}$, satisfying the rank condition
\begin{gather}\label{rankHJ}
\det\left(\frac{\partial u^{(1)}_{\alpha}}{\partial a_{\beta}}\right)\neq 0,
\end{gather}
if and only if the operators (\ref{totaldiff}) pairwise commute, that is
\begin{gather}\label{commutation}
\big[D_{\kbeta},D_{\jalpha}\big]=0
\end{gather}
for all $1\leq \alpha, \beta \leq r$, $1_{\alpha}\leq \jalpha\leq \lalpha$, $1_{\beta}\leq \kbeta\leq \lbeta$. We refer to \cite[Theorem~2.1]{BCR1-2002} for this result. Note that if~(\ref{commutation}) hold, then the Hamilton--Jacobi equation admits locally a solution which is additively separable into groups of variables and satisfies the completeness condition~(\ref{rankPainHJ}) as a consequence of the completeness condition (\ref{rankHJ}). In consequence, such metrics are of the Painlev\'e form~(\ref{Painleveform}).

We now prove that the commutation relations (\ref{commutation}) are equivalent to the generalized Levi-Civita conditions~(\ref{genLeviCivita}). Note that this is a natural generalization of the link between the complete separation of variables which is familiar from the St\"ackel case and the classical Levi-Civita separability conditions, as reviewed in~\cite{BCR1-2002,KaMi1984}. Indeed, we have
\begin{Lemma}
The pairwise commutation relations \eqref{commutation} for the derivations $D_{\kbeta}$ are equivalent to the generalized Levi-Civita conditions~\eqref{genLeviCivita}.
\end{Lemma}

\begin{proof}Substituting the expression (\ref{totaldiff}) of the differential operators $D_{\kbeta}$ into the commutation relations (\ref{commutation}), we obtain
 \begin{gather*}
 \left[-\partial_{\kbeta}\left(\frac{1}{c^{\gamma}}\sum_{\delta =1}^{r}\big(\partial_{\pgamma} c^{\delta}\big)u^{(1)}_{\delta}\right) + \frac{1}{c^{\beta} c^{\gamma}}\sum_{\alpha =1}^{r}\big(\partial_{\kbeta} c^{\alpha}\big)u^{(1)}_{\alpha}\big(\partial_{\pgamma} c^{\beta}\big)\right]\frac{\partial}{\partial u^{(1)}_{\gamma}} \\
 \qquad{} +
 \left[-\partial_{\kbeta}\left(\frac{1}{c^{\gamma}}\sum_{\delta =1}^{r}\big(\partial_{\pgamma} c^{\delta}\big)u^{(1)}_{\delta}\right)+\frac{1}{c^{\beta}c^{\gamma}}\sum_{\alpha =1}^{r}\big(\partial_{\kbeta}c^{\alpha}\big)u^{(1)}_{\alpha}\big(\partial_{\pgamma} c^{\beta}\big)\right]\frac{\partial}{\partial u^{(1)}_{\gamma}}=0.
 \end{gather*}
For $1\leq \gamma \neq \beta \leq r$, the above identity is equivalent to
\begin{gather} \label{GenLC}
\big(\partial_{\kbeta} \log c^{\gamma}\big)\big(\partial_{\pgamma} \log c^{\alpha}\big)+\big(\partial_{\kbeta} \log c^{\alpha}\big)\big(\partial_{\pgamma} \log c^{\beta}\big)-\frac{1}{c^{\alpha}}\partial_{\kbeta}\partial_{\pgamma}c^{\alpha}=0,
\end{gather}
which are precisely the generalized Levi-Civita conditions~(\ref{genLeviCivita}) after relabeling of indices. Note that for $1\leq \gamma = \beta \leq r$, the above identity is always satisfied by a straightforward calculation.
\end{proof}

At this stage, we thus have proved another characterization of Painlev\'e metrics which appears in \cite[p.~12]{CR2019}.

\begin{Proposition} \label{PainLC}A metric $g$ is of Painlev\'e type if and only if there exist block orthogonal coordinates~$\big(\textbf{x}^\alpha\big)$ such that the generalized Levi-Civita conditions \eqref{GenLC} hold.
\end{Proposition}

We finish this section giving still another characterization of Painlev\'e metrics of a more intrinsic nature. The starting point is the observation that the generalized Killing--Eisenhart equations are related to the existence of quadratic first integrals (or symmetries) $K_{(\beta)}$ by the following result proved in~\cite[Proposition~5.3]{CR2019} (see also Lemma~\ref{KE} for an implicit proof of this proposition)

\begin{Proposition}In block orthogonal coordinates, we have that
\[
 \big\{H,K_{(\beta)} \big\} = 0,
\]	
if and only if the Killing--Eisenhart equations
\begin{gather}\label{genKE1}
 \partial_{j_{\gamma}}(\rho_{ \beta \delta})= (\rho_{\beta \gamma} - \rho_{\beta \delta }) \big(\partial_{j_{\gamma}} \log c^\delta\big),
\end{gather}
hold for all $1 \leq \gamma, \delta \leq r$.
\end{Proposition}

The second observation is the fact that the integrability conditions for the Killing--Eisenhart equations~(\ref{genKE1}) are given by
\begin{gather*} 
 (\rho_{\beta \gamma} - \rho_{\beta \delta}) \left[ \big( \partial_{\jgamma} \log c^{\delta} \big) \big( \partial_{\kalpha} \log c^{\gamma}\big) + \big(\partial_{\jgamma} \log c^{\alpha} \big) \big( \partial_{\kalpha} \log c^{\delta} \big) - \frac{1}{c^{\delta}} \partial_{\kalpha}\partial_{\jgamma}c^{\delta} \right] = 0,
\end{gather*}
for all $1 \leq \alpha, \beta, \gamma, \delta \leq r$. Clearly, these integrability conditions can be shortened as
\begin{gather} \label{IC2}
 (\rho_{\beta \gamma} - \rho_{\beta \delta}) \cdot \left[ \textrm{Levi-Civita conditions} \right] = 0.
\end{gather}

Using these two observations, Chanu ad Rastelli proved in \cite[Proposition 5.5]{CR2019} the following characterization of Painlev\'e metrics.

\begin{Proposition} \label{PainKA}
$(M,g)$ is a Painlev\'e manifold if and only if
\begin{enumerate}\itemsep=0pt
\item[$1)$] there exist $r$ independent quadratic first integral $K_{(\beta)}$, $\beta = 1,\dots,r$ such that $K_{(1)} = H$ and $\{ H, K_{(\beta)} \} = 0$.
\item[$2)$] The associated Killing two tensors	$K_{(\beta)}^{ij}$ are simultaneously block-diagonalized and have common normally integrable eigenspaces.
\end{enumerate}
\end{Proposition}
\begin{proof}
If $g$ is a Painlev\'e metric, then the above assertions were already proved.

Assume now that there exist $r$ linearly independent Killing tensors simultaneously in block diagonal form. Then the generalized Killing--Eisenhart equations (\ref{genKE1}) will admit an $r$-dimen\-sio\-nal vector space of solutions. The latter is equivalent to the invertibility of the fundamental matrix $\mathcal{A}$ defined by
\[
{\mathcal{A}}=\left(
\begin{matrix}
\rho_{11}&\dots&\rho_{1r} \\
\vdots& & \vdots \\
\rho_{r1}&\dots&\rho_{rr}
\end{matrix}
\right).
\]
Hence any solution $(\rho_{1},\dots,\rho_{r})$ of (\ref{genKE1}) is given by
\[
\left(
\begin{matrix}\rho_{1}\\
\vdots\\
\rho_{r}
\end{matrix}
\right)
=
{\mathcal{A}}\left(
\begin{matrix}a_{1}\\
\vdots\\
a_{r}
\end{matrix}
\right),
\]
for some constants $(a_{1},\dots,a_{r})\in A$. It is clear then that the Killing--Eisenhart equations are completely integrable and thus satisfy the integrability conditions (\ref{IC2}). Moreover, from the invertibility of~$\mathcal{A}$, we can always choose the constants $(a_1,\dots,a_r)$ such that $\rho_{\alpha} \neq \rho_{\beta}$, $\forall\, 1 \leq \alpha \ne \beta \leq r$ at a point $p\in M$ and therefore in an neighbourhood of $p$, by continuity. Hence, the integrabi\-li\-ty conditions~(\ref{IC2}) reduces to the generalized Levi-Civita separability conditions~(\ref{GenLC}). We conclude that $g$ is a Painlev\'e metric from Proposition~\ref{PainLC}.
\end{proof}

As a concluding remark for this section, we emphasize that the hypotheses~(\ref{H}) that we make to deduce the Levi-Civita conditions from the Killing--Eisenhart equations aren't in fact necessary. Indeed, it follows from Proposition~\ref{PainLC} or Proposition~\ref{PainKA} that the Levi-Civita conditions always hold whenever~$g$ is a Painlev\'e metric, that is a metric of the form~(\ref{Painleveform}).

\section[The generalized Robertson conditions and the separability of the Helmholtz equation]{The generalized Robertson conditions and the separability\\ of the Helmholtz equation}\label{PainHelmSection}

Assume that the manifold $(M,g)$ admits locally block-orthogonal coordinates $\big(\textbf{x}^\alpha\big)$ such that $g = \sum\limits_{\beta = 1}^r c_\beta G_\beta$. Then using the same calculation that led to the expression~(\ref{Laplacian}) of the Laplace--Beltrami operator, the Helmholtz equation~(\ref{HelmPain}) for such a metric reads
\begin{gather}\label{PainHelmeff}
\sum_{\beta =1}^{r}c^{\beta}\left[-\Delta_{G_{\beta}}+\sum_{\jbeta = 1_{\beta}}^{\lbeta}\gamma^{\jbeta}\partial_{\jbeta}\right]u=a_{1}u,
\end{gather}
where $c^\beta = (c_\beta)^{-1}$ and
\[
 \gamma^{\jbeta} = - \sum_{\ibeta = 1_\beta}^{l_\beta} \big(G^\beta \big)^{\ibeta \jbeta} \partial_{\ibeta} \left[ \log\left( \frac{c_{1}^{\frac{l_{1}}{2}}\cdots c_{r}^{\frac{l_{r}}{2}}}{c_{\beta}} \right)\right].
\]

We shall say that a block diagonal metric $g = \sum\limits_{\beta = 1}^r c_\beta G_\beta$ satisfies the generalized Robertson conditions if and only if the differential equations
\begin{gather} \label{RC}
 \partial_{\ialpha} \gamma^{\jbeta} = 0, \qquad \forall\, 1 \leq \alpha \ne \beta \leq r,
\end{gather}
hold.

Note that under the assumption (\ref{RC}), we may write the Helmholtz equation (\ref{PainHelmeff}) as
\[
\sum_{\beta = 1}^{r}c^{\beta} B_{\beta}u=a_{1}u,
\]
where the partial differential operators $B_{\beta}$, $1\leq \beta \leq r,$ defined by
\[
B_{\beta }:=-\Delta_{G_{\beta}}+\sum_{\jbeta = 1_{\beta}}^{\lbeta}\gamma^{\jbeta}\partial_{\jbeta},
\]
now depend on the group of variables ${\bf x}^{\beta}$ only.

In this section, we want to prove

\begin{Proposition} \label{PainSchro}
 Assume that the manifold $(M,g)$ admits locally block-orthogonal coordinates and satisfies the Robertson conditions~\eqref{RC}. Then $g$ is a Painlev\'e metric if and only if there exists a parametrized family of solutions $u$ of the Helmholtz equation~\eqref{PainHelmeff} which is product-separable into groups of variables, of the form
\begin{gather}\label{Blocksep}
u=\prod_{\beta =1}^{r}u_{\beta}\big({\bf x}^{\beta};a_{1},\dots,a_{r}\big),
\end{gather}
and satisfies the completeness condition
\begin{gather}\label{rankPainHelm}
\det \left({\partial_{a_{\alpha}}\left(\frac{B_{\beta} u_{\beta}}{u_{\beta}}\right)}\right)\neq 0.
\end{gather}
\end{Proposition}

\begin{proof}
Assume that $g$ is a Painlev\'e metric. This means that $c_\beta = \frac{\det S}{s^{\beta 1}}$. Then the Helmholtz equation (\ref{PainHelmeff}) may equivalently be written in terms of the generalized St\"ackel matrix $S = (s_{\alpha \beta})$ and a set of arbitrary real parameters $(a_{1}:=\lambda, a_{2},\dots, a_{r})$ defined on an open subset of $A$ of ${\mathbb R}^{r}$ as
\begin{gather} \label{HE1}
\sum_{\beta =1}^{r} c^{\beta} \left[B_{\beta} -\sum_{\alpha =1}^{r}s_{\beta \alpha}a_{\alpha}\right] u = 0.
\end{gather}

We now consider a parametrized family of solutions $u$ of the Helmholtz equation~(\ref{HE1}) which is product-separable into groups of variables, of the form~(\ref{Blocksep}) where for each $1\leq \beta \leq r$, the factor~$u_{\beta}$ is required to satisfy the partial differential equation in the group of variables ${\bf x}^{\beta}$ given by
\begin{gather}\label{SepPDE}
 B_{\beta}u_{\beta} = \left(\sum_{\alpha =1}^{r}s_{\beta \alpha} a_{\alpha}\right) u_\beta.
\end{gather}

We note that for Painlev\'e metrics of Riemannian signature on a compact manifold, it follows by \cite[Theorem~8.3]{GiTru2001} that the elliptic partial differential equation (\ref{SepPDE}) admits a unique solution if the parameters $(a_{1}:=\lambda, a_{2},\dots, a_{r})$ are chosen so that the non-positivity condition
\[
\sum_{\alpha =1}^{r}s_{\beta \alpha} a_{\alpha}\leq 0,
\]
is satisfied (at least locally).

The form of the separated equations (\ref{SepPDE}) and the assumption that the generalized St\"ackel matrix is invertible imply that our parametrized family of block-separable solutions of the Helmholtz equation must satisfy the rank condition (\ref{rankPainHelm}),  i.e.,
\[
\det \left({\partial_{a^{\alpha}}\left(\frac{B_{\beta} u_{\beta}}{u_{\beta}}\right)}\right)\neq 0.
\]

We now prove the converse statement, namely that the existence of a parametrized family of block-separable solutions~(\ref{Blocksep}) of the Helmholtz equation satisfying partial differential equations of the form
\begin{gather}\label{SepB}
\sum_{\beta =1}^{r}c^{\beta}B_{\beta} u=a_{1}u,
\end{gather}
and the rank condition~(\ref{rankPainHelm}) implies that the underlying metric must be of Painlev\'e form. Substituting $u$ of the form (\ref{Blocksep}) into (\ref{SepB}) gives
\[
\sum_{\beta =1}^{r} c^{\beta}\frac{B_{\beta} u_{\beta}}{u_{\beta}}=a_{1}.
\]
Differentiating the latter equation with respect to $a_{\alpha}$, we obtain
\[
\sum_{\beta =1}^{r} c^{\beta}\left({\partial_{a^{\alpha}}\left(\frac{B_{\beta} u_{\beta}}{u_{\beta}}\right)}\right)=\delta^{\alpha}_{1}.
\]
Letting
\[
s_{\beta \alpha}={\partial_{a^{\alpha}}\left(\frac{B_{\beta} u_{\beta}}{u_{\beta}}\right)},
\]
we obtain the expression of $c_{\beta}$ for a Painlev\'e metric as given in (\ref{not}).
\end{proof}

\begin{Remark} From (\ref{SepPDE}) and the fact that the St\"ackel matrix $S$ is invertible, we conclude that the product separable solutions~(\ref{sepform}) satisfy eigenvalue equations of the form
\begin{gather}\label{eigenvalueeq}
T_{\alpha} u = a_{\alpha}u, \qquad 1\leq \alpha \leq r,
\end{gather}
where the $T_{\alpha}$ are the linear second order differential operators given by
\[
 T_{\alpha} = \sum_{\beta = 1}^r \frac{s^{\beta \alpha}}{\det S} B_\beta.
\]
They will be shown in Theorem~\ref{maintheorem2} to be identical to the operators $\Delta_{K_{(\alpha)}}$ defined by~(\ref{KLaplacian}). Hence the separation constants $a_1, \dots, a_r$ can be understood as the natural eigenvalues of the operators~$\Delta_{K_{(\alpha)}}$.
\end{Remark}

\section{Proofs of the main theorems} \label{ProofsofMainTheorems}

\subsection{Proof of Theorem \ref{maintheorem1}}
The fact that the generalized Robertson conditions (\ref{GenRobertson}) are sufficient conditions for the product separability of the Helmholtz equation (\ref{HelmPain}) in the groups of variables associated to a Painlev\'e metric follows from Proposition \ref{PainSchro}.

What remains to be done in order to prove Theorem \ref{maintheorem1} and what constitutes our main task is therefore to show that the generalized Robertson conditions are equivalent to the conditions~(\ref{Ricciform}) on the Ricci tensor, thus generalizing the classical result of Eisenhart \cite{Eis1934} to Painlev\'e metrics. In order to do this, we will show that
\begin{gather}\label{adaptedRicci}
R_{\jalpha \kbeta}=-\frac{3}{4}\partial_{j_{\alpha}}\partial_{k_{\beta}}\log\left[\frac{(\det S)^{n-2}}{\big(s^{11}\big)^{l_{1}}\cdots \big(s^{r1}\big)^{l_{r}}}\right] + \frac{1}{4}T_{\jalpha \kbeta},
\end{gather}
where
\begin{gather}
T_{\jalpha \kbeta} =(l_{\alpha}+l_{\beta}-2)\frac{s^{\alpha 1}s^{\beta 1}}{\det S}\partial_{j_{\alpha}}\partial_{k_{\beta}}\left(\frac{\det S}{s^{\alpha 1}s^{\beta 1}}\right)\nonumber\\
\hphantom{T_{\jalpha \kbeta} =}{} +\sum_{\gamma\neq \alpha,\beta=1}^{r}l_{\gamma}\left[\left(\partial_{j_{\alpha}}\log\frac{s^{\gamma 1}}{\det S}\right)\left(\partial_{k_{\beta}}\log\frac{s^{\alpha 1}}{\det S}\right)\right.\nonumber\\
\left.\hphantom{T_{\jalpha \kbeta} =}{}+\left(\partial_{j_{\alpha}}\log\frac{s^{\beta 1}}{\det S}\right)\left(\partial_{k_{\beta}}\log\frac{s^{\gamma 1}}{\det S}\right)
-\frac{\det S}{s^{\gamma 1}}\partial_{j_{\alpha}}\partial_{k_{\beta}}\left(\frac{s^{\gamma 1}}{\det S}\right)\right].\label{Tterms}
\end{gather}
We note that the expression (\ref{adaptedRicci}) of the off block-diagonal components of the Ricci tensor $R_{\jalpha \kbeta}$ is independent of the $r$ Riemannian block metrics $G_{\beta}$, $1\leq \beta \leq r$ defined by~(\ref{blockmetric}). We also remark that the first term in the expression~(\ref{Tterms}) of~$T_{\jalpha \kbeta}$, involving second derivatives, vanishes identically in the special case of St\"ackel metrics since the pre-factor $\lalpha+\lbeta-2$ is zero in that case.

Once we will have established (\ref{adaptedRicci}), it will then follow from Lemma~\ref{LeviCivita} and more precisely from the generalized Levi-Civita conditions~(\ref{genLeviCivita}) and~(\ref{secondorder}) that the generalized Robertson conditions (\ref{GenRobertson}) are indeed equivalent to the vanishing conditions~(\ref{Ricciform}) on the non-block diagonal components of the Ricci tensor. We therefore proceed to establish the form (\ref{adaptedRicci}) of the Ricci tensor for a Painlev\'e metric.

The expression of the Ricci tensor in terms of the Christoffel symbols is given by
\begin{gather}\label{Ricci}
R_{\jalpha \kbeta}=R^{l}_{{\phantom l} l \jalpha \kbeta}=\partial_{l}{\Gamma^{l}}_{\jalpha \kbeta}-\partial_{\jalpha}{\Gamma^{l}}_{l\kbeta}+{\Gamma^{m}}_{\jalpha \kbeta}{\Gamma^{l}}_{l m}-{\Gamma^{m}}_{l\kbeta}{\Gamma^{l}}_{\jalpha m},
\end{gather}
where the summation convention is applied with $1\leq l,m \leq n = \dim M$. In order to compute the right-hand side of (\ref{Ricci}), we will need expressions for the Christoffel symbols of a Painlev\'e metric (\ref{Painleveform}). Using the standard formulas
\begin{gather*}
\Gamma_{hji}=\frac{1}{2}(\partial_{h}g_{ji}+\partial_{j}g_{ih}-\partial_{i}g_{hj}), \qquad {\Gamma^{i}}_{hj}=g^{ik}\Gamma_{hjk},
\end{gather*}
and writing the Painlev\'e metric (\ref{Painleveform}) in block-diagonal form as
\[
{\rm d}s^{2}=\sum_{\alpha = 1}^{r}\sum_{\ialpha=1_{\alpha}}^{\lalpha}\sum_{\jalpha=1_{\alpha}}^{\lalpha}(g_{\alpha})_{\ialpha \jalpha}{\rm d}x^{\ialpha}{\rm d}x^{\jalpha},
\]
we obtain for fixed indices $1\leq \alpha,\beta \leq r$,
\begin{gather}
{\Gamma^{\ialpha}}_{\kalpha \jbeta}=\frac{1}{2}\sum_{\palpha=1_{\alpha}}^{\lalpha} \big(g^{\alpha}\big)^{\ialpha \palpha}\partial_{\jbeta}(g_{\alpha})_{\kalpha \palpha},\nonumber\\
 {\Gamma^{\jbeta}}_{\ialpha \kalpha}=-\frac{1}{2}\sum_{\pbeta=1_{\beta}}^{\lbeta}\big(g^{\beta}\big)^{\jbeta \pbeta}\partial_{\pbeta}(g_{\alpha})_{\ialpha \kalpha}\qquad \text{for} \quad \alpha \neq \beta,\label{Christoffel}
\end{gather}
and
\begin{gather}\label{Christoffelmixed}
{\Gamma^{\ialpha}}_{\halpha \kalpha}=\frac{1}{2}\sum_{\kalpha=1_{\alpha}}^{\lalpha} g^{\ialpha \kalpha}(\partial_{\halpha}g_{\jalpha \kalpha}+\partial_{\jalpha}g_{\halpha \kalpha}-\partial_{\kalpha}g_{\halpha \jalpha}).
\end{gather}
In view of the expressions (\ref{Christoffel}) of the Christoffel symbols, it is convenient to split the sum over~$l$ appearing in (\ref{Ricci}) into three sums, the first sum corresponding to the values of the summation index $l$ lying in groups of indices different from the groups corresponding to $\alpha$ and $\beta$, and the remaining two to the values of $l$ belonging to the groups of indices labelled by $\alpha$ and $\beta$ respectively. Thus we write
\begin{gather}\label{Riccidecomp}
R_{\jalpha \kbeta}=\sum_{l=1}^{n}R^{l}_{{\phantom l} l \jalpha \kbeta}=\sum_{\gamma \neq \alpha,\, \beta=1}^{r}\sum_{\pgamma=1_{\gamma}}^{\lgamma}R^{\pgamma}_{{\phantom l} \pgamma \jalpha \kbeta} + \sum_{\palpha =1_{\alpha}}^{\lalpha}R^{\palpha}_{{\phantom l} \palpha \jalpha \kbeta}+\sum_{\pbeta=1_{\beta}}^{\lbeta}R^{\pbeta}_{{\phantom l} \pbeta \jalpha \kbeta}.
\end{gather}
Let us begin with the first term. We have, for $\gamma\neq \alpha, \beta$,
\begin{gather}
\sum_{\pgamma = 1_{\gamma}}^{\lgamma}R^{\pgamma}_{{\phantom l} \pgamma \jalpha \kbeta}=-\sum_{\pgamma = 1_{\gamma}}^{\lgamma}\partial_{\jalpha}{\Gamma^{\pgamma}}_{\pgamma \kbeta} +\sum_{\pgamma = 1_{\gamma}}^{\lgamma}\sum_{\malpha=1_{\alpha}}^{\lalpha}{\Gamma^{\malpha}}_{\jalpha \kbeta}{\Gamma^{\pgamma}}_{\pgamma \malpha}\nonumber\\
\hphantom{\sum_{\pgamma = 1_{\gamma}}^{\lgamma}R^{\pgamma}_{{\phantom l} \pgamma \jalpha \kbeta}=}{}
+\sum_{\pgamma = 1_{\gamma}}^{\lgamma}\sum_{\mbeta=1_{\beta}}^{\lbeta}{\Gamma^{\mbeta}}_{\jalpha \kbeta}{\Gamma^{\pgamma}}_{\pgamma \mbeta} -\sum_{\pgamma = 1_{\gamma}}^{\lgamma}\sum_{\mgamma=1_{\gamma}}^{\lgamma}{\Gamma^{\mgamma}}_{\pgamma \kbeta}{\Gamma^{\pgamma}}_{\jalpha \mgamma},\label{Riemann1}
\end{gather}
where the summations have been written out explicitly to avoid notational ambiguities. We begin with evaluating the first-derivative term on the right-hand side of~(\ref{Riemann1}) using the expressions~(\ref{Christoffel}) for the Christoffel symbols, thus obtaining, for $\gamma \neq \alpha, \beta$,
\begin{gather}\label{derivgamma}
\sum_{\pgamma = 1_{\gamma}}^{\lgamma}\partial_{\jalpha}{\Gamma^{\pgamma}}_{\pgamma \kbeta} = \frac{1}{2}\sum_{\pgamma = 1_{\gamma}}^{\lgamma}\sum_{\ngamma = 1_{\gamma}}^{\lgamma}\partial_{\jalpha}\big(\big(g^{\gamma}\big)^{\pgamma \ngamma}\partial_{\kbeta}((g_{\gamma})_{\pgamma \ngamma})\big)=\frac{1}{2}\partial_{\jalpha}\partial_{\kbeta}\log(|g_{\gamma}|),
\end{gather}
where $|g_{\gamma}|:=\det(g_{\igamma \jgamma})$. It follows that
\[
\sum_{\gamma \neq \alpha,\, \beta=1}^{r}\sum_{\pgamma=1_{\gamma}}^{\lgamma}\partial_{\jalpha}{\Gamma^{\pgamma}}_{\pgamma \kbeta}= \frac{1}{2}\partial_{\jalpha}\partial_{\kbeta}\left(\prod_{\gamma \neq \alpha, \,\beta=1}^{r}\log(|g_{\gamma}|)\right).
\]
From the Painlev\'e form (\ref{Painleveform}), we have
\begin{gather}\label{Painleveblock}
g_{\igamma \jgamma}=\frac{\det S}{s^{\gamma 1}}(G_{\gamma})_{\igamma \jgamma},
\end{gather}
so that
\[
|g_{\gamma}|=\frac{(\det S)^{\lgamma}}{\big({s^{\gamma 1}}\big)^{\lgamma}}|G_{\gamma}|.
\]
Using the fact that $|G_{\gamma}|$ is a function of the variables ${\bf x}^{\gamma}$ only, we obtain
\begin{gather*}
\sum_{\gamma \neq \alpha, \,\beta=1}^{r}\sum_{\pgamma=1_{\gamma}}^{\lgamma} \partial_{\jalpha}{\Gamma^{\pgamma}}_{\pgamma \kbeta}=\frac{1}{2}\sum_{\gamma \neq \alpha, \,\beta=1}^{r}\lgamma \partial_{\jalpha}\partial_{\kbeta}\left(\log\left(\frac{\det S}{s^{\gamma 1}}\right)\right).
\end{gather*}
We notice that the above expression is independent of the quadratic differential forms $G_{\alpha}$ defined by~(\ref{blockmetric}).

Next we evaluate the terms quadratic in the Christoffel symbols in the right-hand side of~(\ref{Riemann1}). Again, we use the fact that $\gamma\neq \alpha, \beta$ and the fact that in the Painlev\'e form~(\ref{Painleveform}), each of the quadratic differential forms $G_{\alpha}$ defined by~(\ref{blockmetric}) depends on the group of variables ${\bf x}^{\alpha}$ only. We have
\begin{gather}
\sum_{\malpha=1_{\alpha}}^{\lalpha} \sum_{\pgamma = 1_{\gamma}}^{\lgamma}{\Gamma^{\malpha}}_{\jalpha \kbeta}{\Gamma^{\pgamma}}_{\pgamma \malpha}\nonumber\\
\qquad{} =\frac{1}{4}\sum_{\malpha = 1_{\alpha}}^{l_{\alpha}} \sum_{\pgamma = 1_{\gamma}}^{\lgamma}\sum_{\nalpha=1_{\alpha}}^{l_{\alpha}}(g^{\alpha})^{\malpha \nalpha}\partial_{\kbeta}((g_{\alpha})_{\jalpha \nalpha})g^{\pgamma \hgamma}\partial_{\malpha}((g_{\gamma})_{\pgamma \hgamma})\nonumber\\
\qquad{} =\frac{1}{4}\sum_{\malpha=1_{\alpha}}^{l_{\alpha}}\sum_{\nalpha=1_{\alpha}}^{l_{\alpha}}(g^{\alpha})^{\malpha \nalpha}\partial_{\kbeta}((g_{\alpha})_{\jalpha \nalpha})\partial_{\malpha}\log(|g_{\gamma}|)\nonumber\\
\qquad{} =\frac{1}{4}\sum_{\malpha=1_{\alpha}}^{l_{\alpha}}\sum_{\nalpha=1_{\alpha}}^{l_{\alpha}}\frac{s^{\alpha 1}}{\det S}\big(G^{\alpha}\big)^{\malpha \nalpha}\partial_{\kbeta}\left(\frac{\det S}{s^{\alpha 1}}(G_{\alpha})_{\jalpha \nalpha}\right)\partial_{\malpha}\log(|g_{\gamma}|)\nonumber\\
\qquad{} =\frac{1}{4}\sum_{\malpha=1_{\alpha}}^{l_{\alpha}}\sum_{\nalpha=1_{\alpha}}^{l_{\alpha}}{\delta^{\malpha}}_{\jalpha}\partial_{\kbeta}\log\left(\frac{\det S}{s^{\alpha 1}}\right)\partial_{\malpha}\log \left(\frac{(\det S)^{\lgamma}}{(s^{\gamma 1})^{\lgamma}}\right)\nonumber\\
\qquad{} =\frac{1}{4}\left(\partial_{\kbeta}\log\left(\frac{s^{\alpha 1}}{\det S}\right)\right)\left(\partial_{\jalpha}\log \left(\frac{\big(s^{\gamma 1}\big)^{\lgamma}}{(\det S)^{\lgamma}}\right)\right).\label{quadr1}
\end{gather}
We notice that the above expression is again independent of the quadratic differential forms $G_{\alpha}$ defined by (\ref{blockmetric}). Likewise, we obtain for the next quadratic term in the Christoffel symbols that appears in the right-hand side of~(\ref{Riemann1}),
\begin{gather}\label{quadr2}
\sum_{\mbeta=1_{\beta}}^{\lbeta} \sum_{\pgamma = 1_{\gamma}}^{\lgamma}{\Gamma^{\mbeta}}_{\jalpha \kbeta}{\Gamma^{\pgamma}}_{\pgamma \mbeta}=\frac{1}{4}\left(\partial_{\jalpha}\log\left(\frac{s^{\beta 1}}{\det S}\right)\right)\left(\partial_{\kbeta}\log \left(\frac{\big(s^{\gamma 1}\big)^{\pgamma}}{(\det S)^{\pgamma}}\right)\right).
\end{gather}
For the third and final quadratic term, we have
\begin{align}
\sum_{\mgamma=1_{\gamma}}^{\lgamma}{\Gamma^{\mgamma}}_{\pgamma \kbeta}{\Gamma^{\pgamma}}_{\jalpha \mgamma} &=\frac{1}{4}\sum_{\mgamma=1_{\gamma}}^{\lgamma}{\delta^{\mgamma}}_{\pgamma}\left(\partial_{\kbeta}\log\left(\frac{s^{\gamma 1}}{\det S}\right)\right){\delta^{\pgamma}}_{\mgamma}\left(\partial_{\jalpha}\log\left(\frac{s^{\gamma 1}}{\det S}\right)\right)\nonumber\\
&=\frac{1}{4}\lgamma \left(\partial_{\kbeta}\log\left(\frac{s^{\gamma 1}}{\det S}\right)\right)\left(\partial_{\jalpha}\log\left(\frac{s^{\gamma 1}}{\det S}\right)\right),\label{quadr3}
\end{align}
which is likewise independent of the quadratic differential forms $G_{\alpha}$ defined by (\ref{blockmetric}). Putting together the expressions (\ref{derivgamma}), (\ref{quadr1}), (\ref{quadr2}) and (\ref{quadr3}) we obtain
\begin{gather}
\sum_{1\leq \gamma \neq \alpha, \beta\leq r}R^{\pgamma}_{{\phantom l} \pgamma \jalpha \kbeta}=\sum_{\gamma \neq \alpha, \,\beta =1}^{r}l_{\gamma}\left[\frac{1}{2}\partial_{\jalpha}\partial_{\kbeta}\log\left(\frac{s^{\gamma 1}}{\det S}\right)\right.\nonumber\\
\hphantom{\sum_{1\leq \gamma \neq \alpha, \beta\leq r}R^{\pgamma}_{{\phantom l} \pgamma \jalpha \kbeta}=}{}
+\frac{1}{4}\left(\partial_{\kbeta}\log\left(\frac{s^{\alpha 1}}{\det S}\right)\right)\left(\partial_{\jalpha}\log\left(\frac{s^{\gamma 1}}{\det S}\right)\right)\nonumber\\
\hphantom{\sum_{1\leq \gamma \neq \alpha, \beta\leq r}R^{\pgamma}_{{\phantom l} \pgamma \jalpha \kbeta}=}{}
+ \frac{1}{4}\left(\partial_{\jalpha}\log\left(\frac{s^{\beta 1}}{\det S}\right)\right)\left(\partial_{\kbeta}\log\left(\frac{s^{\gamma 1}}{\det S}\right)\right)\nonumber\\
\left.\hphantom{\sum_{1\leq \gamma \neq \alpha, \beta\leq r}R^{\pgamma}_{{\phantom l} \pgamma \jalpha \kbeta}=}{}
-\frac{1}{4}\left(\partial_{\kbeta}\log\left(\frac{s^{\gamma 1}}{\det S}\right)\right)\left(\partial_{\jalpha}\log\left(\frac{s^{\gamma 1}}{\det S}\right)\right)\right].\label{Riccidecomp1}
\end{gather}
We still need to evaluate the curvature components $R^{\palpha}_{{\phantom l} \palpha \jalpha \kbeta}$ and $R^{\pbeta}_{{\phantom l} \pbeta \jalpha \kbeta}$, which will require a~separate calculation. We have
\begin{gather}
\sum_{\palpha=1_{\alpha}}^{\lalpha}R^{\palpha}_{{\phantom l} \palpha \jalpha \kbeta} =\sum_{\palpha=1_{\alpha}}^{\lalpha}\partial_{\palpha}{\Gamma^{\palpha}}_{\jalpha \kbeta}-\sum_{\palpha=1_{\alpha}}^{\lalpha}\partial_{\jalpha}{\Gamma^{\palpha}}_{\palpha \kbeta}\nonumber\\
\hphantom{\sum_{\palpha=1_{\alpha}}^{\lalpha}R^{\palpha}_{{\phantom l} \palpha \jalpha \kbeta} =}{}
+\sum_{\palpha=1_{\alpha}}^{\lalpha}\sum_{\malpha = 1_{\alpha}}^{\lalpha}{\Gamma^{\malpha}}_{\jalpha \kbeta}{\Gamma^{\palpha}}_{\palpha \malpha}+\sum_{\palpha=1_{\alpha}}^{\lalpha}\sum_{\mbeta=1_{\beta}}^{\lbeta}{\Gamma^{\mbeta}}_{\jalpha \kbeta}{\Gamma^{\palpha}}_{\palpha \mbeta}\nonumber\\
\hphantom{\sum_{\palpha=1_{\alpha}}^{\lalpha}R^{\palpha}_{{\phantom l} \palpha \jalpha \kbeta} =}{}-\sum_{\palpha=1_{\alpha}}^{\lalpha}\sum_{\malpha=1_{\alpha}}^{\lalpha}{\Gamma^{\malpha}}_{\palpha \kbeta}{\Gamma^{\palpha}}_{\jalpha \malpha}-\sum_{\palpha=1_{\alpha}}^{\lalpha}\sum_{\mbeta=1_{\beta}}^{\lbeta}{\Gamma^{\mbeta}}_{\lalpha \kbeta}{\Gamma^{\palpha}}_{\jalpha \mbeta},\label{Riemann2}
\end{gather}
where again we have written out the summation signs explicitly to avoid notational ambiguities. We have, using (\ref{Christoffel}),
\begin{gather}\label{moreChristoffel}
{\Gamma^{\palpha}}_{\jalpha \kbeta}=\frac{1}{2}{\delta^{\palpha}}_{\jalpha}\partial_{\kbeta}\left(\log\left(\frac{\det S}{s^{\alpha 1}}\right)\right),
\end{gather}
so using the fact that the cofactor $s^{\alpha 1}$ is independent of ${\bf x}^{\alpha}$, we obtain
\begin{gather}\label{newderiv1}
\sum_{\palpha=1_{\alpha}}^{\lalpha}\partial_{\palpha}{\Gamma^{\palpha}}_{\jalpha \kbeta}=\frac{1}{2}\partial_{\jalpha}\partial_{\kbeta}\big(\log(\det S)\big),
\end{gather}
and a similar calculation gives
\begin{gather}\label{newderiv2}
\sum_{\palpha=1_{\alpha}}^{\lalpha}\partial_{\jalpha}{\Gamma^{\palpha}}_{\palpha \kbeta}=\frac{1}{2}\partial_{\jalpha}\partial_{\kbeta}\big(\log\big((\det S)^{\lalpha}\big)\big).
\end{gather}
We now evaluate the quadratic terms in the Christoffel symbols that appear in the curvature component~(\ref{Riemann2}). In order to do so, we substitute into the expression (\ref{Christoffelmixed}) of the Christoffel symbols the expressions
\[
(g_{\alpha})_{\ialpha \jalpha}=\frac{\det S}{s^{\alpha 1}}G_{\alpha},
\]
which result from the Painlev\'e form (\ref{Painleveform}). We obtain the following expressions for the Christoffel symbols,
\begin{gather}
{\Gamma^{\ialpha}}_{\halpha \jalpha}={\gamma^{\ialpha}}_{\halpha \jalpha} +\frac{1}{2}\Bigg[ {\delta^{\ialpha}}_{\jalpha}\partial_{\halpha}\big(\log{(\det S)}\big)+{\delta^{\ialpha}}_{\halpha}\partial_{\jalpha}\big(\log{(\det S)}\big)\nonumber\\
\hphantom{{\Gamma^{\ialpha}}_{\halpha \jalpha}=}{}-\sum_{\kalpha = 1_{\alpha}}^{\lalpha}\big(G^{\alpha}\big)^{\ialpha \kalpha}(G_{\alpha})_{\halpha \jalpha}\partial_{\kalpha}\big(\log{(\det S)}\big)\Bigg],\label{longChristoffel}
\end{gather}
where the ${\gamma^{\ialpha}}_{\halpha \jalpha}$ denote the Christoffel symbols of the block metric~$G_{\alpha}$ given by~(\ref{blockmetric}). It then follows from~(\ref{longChristoffel}) and~(\ref{moreChristoffel}) that
\begin{gather}
\sum_{\palpha=1_{\alpha}}^{\lalpha}\sum_{\malpha = 1_{\alpha}}^{\lalpha}{\Gamma^{\malpha}}_{\jalpha \kbeta}{\Gamma^{\palpha}}_{\palpha \malpha}\nonumber\\
\qquad{} =\frac{1}{2}\partial_{\kbeta}\left(\log\left(\frac{\det S}{s^{\alpha 1}}\right)\right)\left[\sum_{\palpha=1_{\alpha}}^{\lalpha}{\gamma^{\palpha}}_{\palpha \jalpha}+\frac{1}{2}\lalpha \partial_{\malpha}\big(\log{(\det S)}\big)\right].\label{newquadr1}
\end{gather}
Similarly, using (\ref{moreChristoffel}) and (\ref{longChristoffel}), we obtain
\begin{gather}
\sum_{\palpha=1_{\alpha}}^{\lalpha}\sum_{\malpha=1_{\alpha}}^{\lalpha}{\Gamma^{\malpha}}_{\palpha \kbeta}{\Gamma^{\palpha}}_{\jalpha \malpha}\nonumber\\
\qquad{} =\frac{1}{2}\partial_{\kbeta}\left(\log\left(\frac{\det S}{s^{\alpha 1}}\right)\right)\left[\sum_{\palpha=1_{\alpha}}^{\lalpha}{\gamma^{\palpha}}_{\jalpha \palpha}+\frac{1}{2}\lalpha \partial_{\malpha}\big(\log{(\det S)}\big)\right].\label{newquadr2}
\end{gather}
It follows therefore from (\ref{newquadr1}) and (\ref{newquadr2}) that
\[
\sum_{\palpha=1_{\alpha}}^{\lalpha}\sum_{\malpha = 1_{\alpha}}^{\lalpha}{\Gamma^{\malpha}}_{\jalpha \kbeta}{\Gamma^{\palpha}}_{\palpha \malpha}-\sum_{\palpha=1_{\alpha}}^{\lalpha}\sum_{\malpha=1_{\alpha}}^{\lalpha}{\Gamma^{\malpha}}_{\palpha \kbeta}{\Gamma^{\palpha}}_{\jalpha \malpha}=0.
\]
We now evaluate the remaining difference of two double sums in the expression (\ref{Riemann2}) of the curvature, using the expressions
\[
{\Gamma^{\mbeta}}_{\jalpha \kbeta}=\frac{1}{2}{\delta^{\mbeta}}_{\kbeta}\partial_{\jalpha}\left(\log\left(\frac{\det S}{s^{\beta 1}}  \right)\right), \qquad {\Gamma^{\mbeta}}_{\jalpha \kbeta} = \frac{1}{2}\partial_{\mbeta}\left(\log\left(\frac{(\det S)^{\lalpha}}{\big(s^{\beta 1}\big)^{\lalpha}} \right)\right),
\]
for the Christoffel symbols, to obtain
\begin{gather}
\sum_{\palpha=1_{\alpha}}^{\lalpha}\sum_{\malpha = 1_{\alpha}}^{\lalpha}{\Gamma^{\malpha}}_{\jalpha \kbeta}{\Gamma^{\palpha}}_{\palpha \malpha}-\sum_{\palpha=1_{\alpha}}^{\lalpha}\sum_{\mbeta=1_{\beta}}^{\lbeta}{\Gamma^{\mbeta}}_{\lalpha \kbeta}{\Gamma^{\palpha}}_{\jalpha \mbeta}\nonumber\\
\qquad{} =\frac{1}{4}(\lalpha-1)\partial_{\jalpha}\left(\log \left(\frac{\det S}{s^{\beta 1}}\right)\right)\partial_{\kbeta}\left(\log\left(\frac{\det S}{s^{\alpha 1}} \right)\right).\label{newquadr3}
\end{gather}
Substituting the expressions (\ref{newderiv1}), (\ref{newderiv2}), (\ref{newquadr1}), (\ref{newquadr2}) and (\ref{newquadr3}) into (\ref{Riemann2}), we obtain
\begin{gather}
\sum_{\palpha=1_{\alpha}}^{\lalpha}R^{\palpha}_{{\phantom l} \palpha \jalpha \kbeta}= \frac{1-\lalpha}{2} \partial_{\jalpha}\partial_{\kbeta} \big(\log(\det S)\big)\nonumber\\
\hphantom{\sum_{\palpha=1_{\alpha}}^{\lalpha}R^{\palpha}_{{\phantom l} \palpha \jalpha \kbeta}=}{}
+\frac{1}{4}(\lalpha-1)\partial_{\jalpha}\left(\log\left(\frac{\det S}{s^{\beta 1}}\right)\right)\partial_{\kbeta}\left(\log\left(\frac{\det S}{s^{\alpha 1}}\right)\right),\label{Riccidecomp2}
\end{gather}
and similarly
\begin{gather}\label{Riccidecomp3}
\sum_{\pbeta=1_{\beta}}^{\lalpha}R^{\pbeta}_{{\phantom l} \pbeta \jalpha \kbeta}=\frac{1-\lbeta}{2}\partial_{\jalpha}\partial_{\kbeta} \big(\log(\det S)\big)\nonumber\\
\hphantom{\sum_{\pbeta=1_{\beta}}^{\lalpha}R^{\pbeta}_{{\phantom l} \pbeta \jalpha \kbeta}=}{}
+\frac{1}{4}(\lbeta-1)\partial_{\jalpha}\left(\log\left(\frac{\det S}{s^{\beta 1}}\right)\right)\partial_{\kbeta}\left(\log\left(\frac{\det S}{s^{\alpha 1}}\right)\right).
\end{gather}
We now substitute the expressions (\ref{Riccidecomp1}), (\ref{Riccidecomp2}) and (\ref{Riccidecomp3}) into the decomposition~(\ref{Riccidecomp}) of the off-block diagonal Ricci curvature components $R_{\jalpha \kbeta}$ to obtain, for $\alpha \neq \beta$,
\begin{gather}
R_{\jalpha \kbeta}=-\frac{1}{2}\sum_{\gamma =1}^{r}\lgamma \partial_{\jalpha}\partial_{\kbeta}\left(\log\left(\frac{\det S}{s^{\gamma 1}}\right)\right)+\partial_{\jalpha}\partial_{\kbeta}\big(\log(\det S)\big)\nonumber\\
\hphantom{R_{\jalpha \kbeta}=}{} + \frac{1}{4}(\lalpha+\lbeta-2) \partial_{\jalpha}\left(\log\left(\frac{\det S}{s^{\beta 1}}\right)\right)\partial_{\kbeta}\left(\log\left(\frac{\det S}{s^{\alpha 1}}\right)\right)\nonumber\\
\hphantom{R_{\jalpha \kbeta}=}{}+\frac{1}{4}\sum_{\gamma \neq \alpha, \,\beta=1}^{r}\lgamma \left[\partial_{\kbeta}\left(\log\left(\frac{\det S}{s^{\alpha 1}}\right)\right) \partial_{\jalpha}\left(\log\left(\frac{\det S}{s^{\gamma 1}}\right)\right)\right.\nonumber\\
\hphantom{R_{\jalpha \kbeta}=}{}
+\partial_{\jalpha}\left(\log\left(\frac{\det S}{s^{\beta 1}}\right)\right)\partial_{\kbeta}\left(\log\left(\frac{\det S}{s^{\gamma 1}}\right)\right)\nonumber\\
\left.\hphantom{R_{\jalpha \kbeta}=}{} -\partial_{\kbeta}\left(\log\left(\frac{\det S}{s^{\gamma 1}}\right)\right)\partial_{\jalpha}\left(\log\left(\frac{\det S}{s^{\gamma 1}}\right)\right)\right].\label{Riccialmostfinal}
\end{gather}
We observe that since the cofactors $s^{\alpha 1}$ and $s^{\beta 1}$ are independent of the groups of variables ${\bf x}^{\alpha}$ and ${\bf x}^{\beta}$ respectively and since $\sum\limits_{\alpha=1}^{r}\lalpha = n$, we have
\begin{gather*}
-\frac{1}{2}\sum_{\gamma =1}^{r}\lgamma \partial_{\jalpha}\partial_{\kbeta}\left(\log\left(\frac{\det S}{s^{\gamma 1}}\right)\right)+\partial_{\jalpha}\partial_{\kbeta}\big(\log(\det S)\big)\\
\qquad{} =-\frac{1}{2}\partial_{j_{\alpha}}\partial_{k_{\beta}}\log\left[\frac{(\det S)^{n-2}}{\big(s^{11}\big)^{l_{1}}\cdots \big(s^{r1}\big)^{l_{r}}}\right].
\end{gather*}
We write
\begin{gather}
-\frac{1}{2}\partial_{j_{\alpha}}\partial_{k_{\beta}}\log\left[\frac{(\det S)^{n-2}}{\big(s^{11}\big)^{l_{1}}\cdots \big(s^{r1}\big)^{l_{r}}}\right]\nonumber\\
\qquad{} =-\frac{3}{4}\partial_{j_{\alpha}}\partial_{k_{\beta}}\log\left[\frac{(\det S)^{n-2}}{\big(s^{11}\big)^{l_{1}}\cdots \big(s^{r1}\big)^{l_{r}}}\right] +\frac{1}{4}\partial_{j_{\alpha}}\partial_{k_{\beta}}\log\left[\frac{(\det S)^{n-2}}{\big(s^{11}\big)^{l_{1}}\cdots \big(s^{r1}\big)^{l_{r}}}\right],\label{simplif}
\end{gather}
and compute the second term in (\ref{simplif}) as follows
\begin{gather*}
\frac{1}{4}\partial_{j_{\alpha}}\partial_{k_{\beta}}\log\left[\frac{(\det S)^{n-2}}{\big(s^{11}\big)^{l_{1}}\cdots \big(s^{r1}\big)^{l_{r}}}\right]\\
\qquad{} =\frac{1}{4}\partial_{j_{\alpha}}\partial_{k_{\beta}}\left[ \sum_{1\leq \gamma \neq \alpha, \beta\leq r} \log\left( \frac{(\det S)^{\lgamma}}{\big(s^{\gamma 1}\big)^{\lgamma}}\right)+\log\left(\frac{(\det S)^{\lalpha+\lbeta-2}}{s^{\alpha 1}s^{\beta 1}}\right)\right].
\end{gather*}
Evaluating the derivatives an using the fact that the cofactors $s^{\alpha 1}$ and $s^{\beta 1}$ are independent of the groups of variables ${\bf x}^{\alpha}$ and ${\bf x}^{\alpha}$, we conclude that the expression (\ref{Riccialmostfinal}) of the off-block diagonal Ricci curvature can be written as
\begin{gather*}
R_{\jalpha \kbeta} =-\frac{3}{4}\partial_{j_{\alpha}}\partial_{k_{\beta}}\log\left[\frac{(\det S)^{n-2}}{\big(s^{11}\big)^{l_{1}}\cdots \big(s^{r1}\big)^{l_{r}}}\right] + \frac{1}{4} (l_{\alpha}+l_{\beta}-2)\frac{s^{\alpha 1}s^{\beta 1}}{\det S}\partial_{j_{\alpha}}\partial_{k_{\beta}}\left(\frac{\det S}{s^{\alpha 1}s^{\beta 1}}\right)\\
\hphantom{R_{\jalpha \kbeta} =}{}+
\sum_{\gamma\neq \alpha,\beta=1}^{r}l_{\gamma}\left(\partial_{j_{\alpha}}\log\frac{s^{\gamma 1}}{\det S}\right)\left(\partial_{k_{\beta}}\log\frac{s^{\alpha 1}}{\det S}\right)
+\left(\partial_{j_{\alpha}}\log\frac{s^{\beta 1}}{\det S}\right)\left(\partial_{k_{\beta}}\log\frac{s^{\gamma 1}}{\det S}\right)\\
\hphantom{R_{\jalpha \kbeta} =}{}
-\frac{\det S}{s^{\gamma 1}}\partial_{\jalpha}\partial_{\kbeta}\left(\frac{s^{\gamma 1}}{\det S}\right).
\end{gather*}
This completes the proof of Theorem \ref{maintheorem1}.

\subsection{Proof of Theorem \ref{maintheorem2}}
Our proof follows the structure of the one given in \cite{BCR2-2002} for the special case of Theorem~\ref{maintheorem2} corresponding to St\"ackel metrics satisfying the classical Robertson conditions~(\ref{classicalRobertson}). We shall begin from the general expression for the commutator of two operators of the form
\[
\Delta_{K_{(\alpha)}}=\sum_{\gamma=1}^{r}\sum_{i_{\gamma}=1_{\gamma}}^{l_{\gamma}}\sum_{j_{\gamma}=1_{\gamma}}^{l_{\gamma}}A_{(\alpha)}^{\igamma \jgamma}\partial_{\igamma}\partial_{\jgamma}+\sum_{\gamma=1}^{r}\sum_{j_{\gamma}=1_{\gamma}}^{l_{\gamma}}B_{(\alpha)}^{\jgamma}\partial_{\jgamma},
\]
and analyze this expression for the case where $\Delta_{K_{(\alpha)}}$ is given by
\[
\Delta_{K_{(\alpha)}}=\nabla_{i}(K_{(\alpha)}^{ij}\nabla_{j})=\sum_{\gamma=1}^{r}\sum_{i_{\gamma}=1_{\gamma}}^{l_{\gamma}}\sum_{j_{\gamma} =1_{\gamma}}^{l_{\gamma}}\nabla_{i_{\gamma}}\big(K_{(\alpha)}^{i_{\gamma}j_{\gamma}}\nabla_{j_{\gamma}}\big),
\]
where the $\big(K_{(\alpha)}^{ij}\big)$ are the Killing tensors defined by (\ref{Killingtensor}). We shall then prove that the commutator is identically zero for all Painlev\'e metrics satisfying the generalized Robertson conditions. We shall see that the generalized Killing--Eisenhart equations (\ref{genKE}) established in Lemma \ref{KE} play a key role in the analysis of the commutator.

We have
\begin{gather}
\big[\Delta_{K_{(\alpha)}},\Delta_{K_{(\beta)}}\big] =\sum_{\gamma,\epsilon=1}^{r}\sum_{\igamma = 1_{\gamma}}^{\lgamma}\sum_{\jgamma = 1_{\gamma}}^{\lgamma}\sum_{\kepsilon =1_{\epsilon}}^{\lepsilon}\sum_{\pepsilon =1_{\epsilon}}^{\lepsilon}\Big\{2\big(A_{(\alpha)}^{\igamma \jgamma}\partial_{\igamma}A_{(\beta)}^{\kepsilon \pepsilon}-A_{(\beta)}^{\igamma \jgamma}\partial_{\igamma}A_{(\alpha)}^{\kepsilon \pepsilon}\big)\partial_{\jgamma}\partial_{\kepsilon}\partial_{\pepsilon}\nonumber\\
\hphantom{\big[\Delta_{K_{(\alpha)}},\Delta_{K_{(\beta)}}\big] =}{}
+\big(A_{(\alpha)}^{\igamma \jgamma}\partial_{\igamma}\partial_{\jgamma}A_{(\beta)}^{\kepsilon \pepsilon}-A_{(\beta)}^{\igamma \jgamma}\partial_{\igamma}\partial_{\jgamma}A_{(\alpha)}^{\kepsilon \pepsilon}+B_{(\alpha)}^{\jbeta}\partial_{\jgamma}A_{(\beta)}^{\kepsilon \pepsilon}\nonumber\\
\hphantom{\big[\Delta_{K_{(\alpha)}},\Delta_{K_{(\beta)}}\big] =}{}
-B_{(\beta)}^{\jbeta}\partial_{\jgamma}A_{(\alpha)}^{\kepsilon \pepsilon}\big)\partial_{\kepsilon}\partial_{\pepsilon}
+2\big(A_{(\alpha)}^{\igamma \jgamma}\partial_{\igamma}B_{\beta}^{\pepsilon}-A_{(\beta)}^{\igamma \jgamma}\partial_{\igamma}B_{\alpha}^{\pepsilon}\big)\partial_{\jgamma}\partial_{\pepsilon}\nonumber\\
\hphantom{\big[\Delta_{K_{(\alpha)}},\Delta_{K_{(\beta)}}\big] =}{}
+\big(A_{(\alpha)}^{\igamma \jgamma}\partial_{\igamma}\partial_{\jgamma}B_{(\beta)}^{\pepsilon}-A_{(\beta)}^{\igamma \jgamma}\partial_{\igamma}\partial_{\jgamma}B_{(\alpha)}^{\pepsilon}+B_{(\alpha)}^{\jgamma} \partial_{\jgamma}B_{(\beta)}^{\pepsilon}\nonumber\\
\hphantom{\big[\Delta_{K_{(\alpha)}},\Delta_{K_{(\beta)}}\big] =}{}
-B_{(\beta)}^{\jgamma}\partial_{\jgamma}B_{(\alpha)}^{\pepsilon}\big)\partial_{\pepsilon}\Big\}.
\label{rawcomm}
\end{gather}
We will now evaluate the coefficients of the third, second and first derivatives in the expres\-sion~(\ref{rawcomm}) of the commutator for Painlev\'e metrics satisfying the generalized Robertson conditions, and show that must vanish identically.

Over the course of the calculations, it will be useful to rewrite the expressions~(\ref{Killingtensor}) of the block components of the Killing tensors $\big(K_{(\alpha)}^{ij}\big)$ in the form
\begin{gather}\label{Killingtensormod}
K_{(\alpha)}^{\igamma \jgamma}=\rho_{\alpha \gamma}\frac{s^{\gamma 1}}{\det S}\big(G^{\gamma}\big)^{\igamma \jgamma}=\rho_{\alpha \gamma}\big(g^{\gamma}\big)^{\igamma \jgamma},
\end{gather}
where the quantities $\rho_{\alpha \gamma}$ are defined by
\begin{gather}\label{rhotwo}
\rho_{\alpha \gamma}=\frac{s^{\gamma \alpha}}{s^{\gamma 1}},
\end{gather}
in which case we obtain
\begin{gather}
A_{(\alpha)}^{\igamma \jgamma}=K_{(\alpha)}^{\igamma \jgamma}=\rho_{\alpha \gamma}\frac{s^{\gamma 1}}{\det S}\big(G^{\gamma}\big)^{\igamma \jgamma}=\rho_{\alpha \gamma}\big(g^{\gamma}\big)^{\igamma \jgamma},\nonumber\\
 B_{(\alpha)}^{\jgamma}=-\rho_{\alpha \gamma}\sum_{\igamma = 1_{\gamma}}^{\lgamma} \big(g^{\gamma}\big)^{\igamma \jgamma}\Gamma_{\igamma},\label{coeffmod}
\end{gather}
where
\[
\Gamma_{\kepsilon}=-\frac{1}{2}\partial_{\kepsilon}\log (|g|)-\sum_{\pepsilon=1_{\epsilon}}^{\lepsilon}\sum_{\hepsilon=1_{\epsilon}}^{\lepsilon}(g_{\epsilon})_{\kepsilon \pepsilon}\partial_{\hepsilon}(g^{\epsilon})^{\hepsilon \pepsilon}.
\]
Note that a standard calculation shows the important result:

\begin{Lemma} \label{RC-NewFormulation}The generalized Robertson conditions \eqref{GenRobertson} are equivalent to
\begin{gather*} 
 \partial_{\jalpha} \Gamma_{\kbeta} = 0, \qquad \forall\, 1 \leq \alpha \ne \beta \leq r.
\end{gather*}
\end{Lemma}

We will compute the expressions of the derivatives of the coefficients $A_{(\alpha)}^{\igamma \jgamma}$ and $B_{(\alpha)}^{\jgamma}$ when needed during the calculations, making use of the generalized Killing--Eisenhart equations~(\ref{genKE}).

We begin with the coefficients of the third derivatives in (\ref{rawcomm}), whose vanishing is equivalent to the condition
\begin{gather}\label{Thierry(a)}
\sum_{\igamma = 1_{\gamma}}^{\lgamma}\big(A_{(\alpha)}^{\igamma \jgamma}\partial_{\igamma}A_{(\beta)}^{\kepsilon \pepsilon}-A_{(\beta)}^{\igamma \jgamma}\partial_{\igamma}A_{(\alpha)}^{\kepsilon \pepsilon}\big)=0.
\end{gather}
We shall see shortly that in analogy with the St\"ackel case, the vanishing of these coefficients does not require the generalized Robertson conditions and holds true for all Painlev\'e metrics. When the expression~(\ref{Killingtensor}) of the Killing tensors $\big(K_{(\alpha)}^{ij}\big)$ is substituted into the condition~(\ref{Thierry(a)}), the latter reduces to
\begin{gather}
\sum_{\igamma = 1_{\gamma}}^{\lgamma}\rho_{\alpha \gamma}\frac{s^{\gamma 1}}{\det S}G^{\igamma \jgamma}\left(\partial_{\igamma}\left(\rho_{\beta \epsilon}\frac{s^{\epsilon 1}}{\det S}\right)\big(G^{\epsilon}\big)^{\kepsilon \pepsilon}+\rho_{\beta \epsilon}\frac{s^{\epsilon 1}}{\det S}\partial_{\igamma}\big(G^{\epsilon}\big)^{\kepsilon \pepsilon}\right)\nonumber\\
\qquad {} = \sum_{\igamma = 1_{\gamma}}^{\lgamma}\rho_{\beta \gamma}\frac{s^{\gamma 1}}{\det S}G^{\igamma \jgamma}\left(\partial_{\igamma}\left(\rho_{\alpha \epsilon}\frac{s^{\epsilon 1}}{\det S}\right)\big(G^{\epsilon}\big)^{\kepsilon \pepsilon}+\rho_{\alpha \epsilon}\frac{s^{\epsilon 1}}{\det S}\partial_{\igamma}\big(G^{\epsilon}\big)^{\kepsilon \pepsilon}\right).\label{3rdorder}
\end{gather}
We now distinguish between the cases $\gamma=\epsilon$ and $\gamma\neq \epsilon$ when analyzing (\ref{3rdorder}). If $\gamma=\epsilon$, then the derivatives of $G^{\epsilon}$ cancel out in (\ref{3rdorder}) and using (\ref{rhotwo}), the condition (\ref{3rdorder}) reduces to
\[
\sum_{\igamma = 1_{\gamma}}^{\lgamma} \big(G^{\gamma}\big)^{\igamma \jgamma}\big(G^{\gamma}\big)^{\kgamma \pgamma}\left( \frac{s^{\gamma \alpha}}{\det S}\partial_{\igamma}\left(\frac{s^{\gamma \beta}}{\det S}\right)-\frac{s^{\gamma \beta}}{\det S}\partial_{\igamma}\left(\frac{s^{\gamma \alpha}}{\det S}\right)\right)=0.
\]
But the latter is an identity on account of the generalized Killing--Eisenhart equations expressed in the form~(\ref{Stackelcomm}). Likewise, if $\gamma \neq \epsilon$, then using the fact that we have then $\partial_{\igamma}\big(G^{\epsilon}\big)^{\kepsilon \pepsilon}=0$, the condition~(\ref{3rdorder}) becomes
\[
\sum_{\igamma = 1_{\gamma}}^{\lgamma}\big(G^{\gamma}\big)^{\igamma \jgamma}\big(G^{\epsilon}\big)^{\kepsilon \pepsilon}\left( \frac{s^{\gamma \alpha}}{\det S}\partial_{\igamma}\left(\frac{s^{\epsilon \beta}}{\det S}\right)-\frac{s^{\gamma \beta}}{\det S}\partial_{\igamma}\left(\frac{s^{\epsilon \alpha}}{\det S}\right)\right)=0,
\]
which is again an identity by virtue of the generalized Killing--Eisenhart equations expressed in the form (\ref{Stackelcomm}).

We now proceed with the coefficients of the second derivatives in~(\ref{rawcomm}), whose vanishing is equivalent to the conditions
\begin{gather}
\sum_{\igamma=!_{\gamma}}^{\lgamma}\sum_{\jgamma=1_{\gamma}}^{\lgamma}\big(A_{(\alpha)}^{\igamma \jgamma}\partial_{\igamma}\partial_{\jgamma}\big(A_{(\beta)}^{\kepsilon \pepsilon}\big)-A_{(\beta)}^{\igamma \jgamma}\partial_{\igamma}\partial_{\jgamma}\big(A_{(\alpha)}^{\kepsilon \pepsilon}\big)+B_{(\alpha)}^{\jgamma}\partial_{\jgamma}\big(A_{(\beta)}^{\kepsilon \pepsilon}\big)-B_{(\beta)}^{\jgamma}\partial_{\jgamma}\big(A_{(\alpha)}^{\kepsilon \pepsilon}\big)\big)\nonumber\\
\qquad{} +2\sum_{\iepsilon=!_{\epsilon}}^{\lepsilon}\big(A_{(\alpha)}^{\iepsilon (\kepsilon}\partial_{\iepsilon}\big(B_{(\beta)}^{\pepsilon )}\big)-A_{(\beta)}^{\iepsilon (\kepsilon}\partial_{\iepsilon}\big(B_{(\alpha)}^{\pepsilon )}\big)\big)=0,\label{2ndorder}
\end{gather}
and
\begin{gather}\label{2ndordermore}
\sum_{\igamma=!_{\gamma}}^{\lgamma}\big(A_{(\alpha)}^{\igamma (\jgamma}\partial_{\igamma}\big(B_{(\beta)}^{\pepsilon )}\big)-A_{(\beta)}^{\igamma (\jgamma}\partial_{\igamma}\big(B_{(\alpha)}^{\pepsilon )}\big)\big)=0,
\end{gather}
where the round brackets denote symmetrization of indices.

We begin with the condition (\ref{2ndorder}), for which we will need expressions for the first and second derivatives of $A_{(\alpha)}^{\igamma \jgamma}$ and $B_{(\alpha)}^{\jgamma}$. Using (\ref{Killingtensormod}) and (\ref{coeffmod}), we obtain
\[
\partial_{\kepsilon}A_{(\alpha)}^{\igamma \jgamma}=\left(\partial_{\kepsilon}(\rho_{\alpha \gamma}) \frac{s^{\gamma 1}}{\det S}+\rho_{\alpha \gamma}\partial_{\kepsilon}\left(\frac{s^{\gamma 1}}{\det S}\right)\right)\big(G^{\gamma}\big)^{\igamma \jgamma}+\rho_{\alpha \gamma}\frac{s^{\gamma 1}}{\det S}\partial_{\kepsilon}\big(\big(G^{\gamma}\big)^{\igamma \jgamma}\big){\delta^{\gamma}}_{\epsilon}.
\]
Using the generalized Killing--Eisenhart equations (\ref{genKE}), the preceding equation reduces to
\begin{gather}\label{firstA}
\partial_{\kepsilon}A_{(\alpha)}^{\igamma \jgamma}=\rho_{\alpha \epsilon}\partial_{\kepsilon}\left(\frac{s^{\gamma 1}}{\det S}\right)\big(G^{\gamma}\big)^{\igamma \jgamma}\!+\rho_{\alpha \gamma}\frac{s^{\gamma 1}}{\det S}\partial_{\kepsilon}\big(\big(G^{\gamma}\big)^{\igamma \jgamma}\big){\delta^{\gamma}}_{\epsilon}
=\rho_{\alpha \epsilon}\partial_{\kepsilon}\big(\big(g^{\gamma}\big)^{\igamma \jgamma}\big).\!\!\!\!
\end{gather}
Likewise, we obtain the following expressions for the second derivatives of the coefficients $A_{(\alpha)}^{\igamma \jgamma}$ by using again the generalized Killing--Eisenhart equations~(\ref{genKE})
\begin{gather}\label{secondA}
\partial_{\kepsilon}\partial_{\pepsilon}A_{(\alpha)}^{\igamma \jgamma}=\rho_{\alpha \epsilon}\partial_{\kepsilon}\partial_{\pepsilon}\big(\big(g^{\gamma}\big)^{\igamma \jgamma}\big).
\end{gather}
This implies immediately that the first two terms in (\ref{2ndorder}) cancel each other out, that is
\begin{gather}
\sum_{\igamma=1_{\gamma}}^{\lgamma}\sum_{\jgamma=1_{\gamma}}^{\lgamma}\big(A_{(\alpha)}^{\igamma \jgamma}\partial_{\igamma}\partial_{\jgamma}\big(A_{(\beta)}^{\kepsilon \pepsilon}\big)-A_{(\beta)}^{\igamma \jgamma}\partial_{\igamma}\partial_{\jgamma}\big(A_{(\alpha)}^{\kepsilon \pepsilon}\big)\big)\nonumber\\
\qquad{} =\sum_{\igamma=!_{\gamma}}^{\lgamma}\sum_{\jgamma=1_{\gamma}}^{\lgamma}\big[\rho_{\alpha \gamma} \big(g^{\gamma}\big)^{\igamma \jgamma}\rho_{\beta \gamma}\partial_{\igamma}\partial_{\jgamma}(g^{\epsilon})^{\kepsilon \pepsilon}-\rho_{\beta \gamma} \big(g^{\gamma}\big)^{\igamma \jgamma}\rho_{\alpha \gamma}\partial_{\igamma}\partial_{\jgamma}(g^{\epsilon})^{\kepsilon \pepsilon}\big]=0.\!\!\!\!\label{cancel}
\end{gather}
By making use of the preceding remark, by observing that
\begin{gather}\label{BintermsofA}
B_{(\alpha)}^{\jgamma}= -\sum_{\igamma = 1_{\gamma}}^{\lgamma}A_{(\alpha)}^{\igamma \jgamma}\Gamma_{\igamma},
\end{gather}
and by using the expressions (\ref{firstA}) and (\ref{secondA}) for the first and second derivatives of the coefficients $A_{(\alpha)}^{\igamma \jgamma}$, the condition (\ref{secondorder}) becomes
\begin{gather*}
\sum_{\iepsilon = 1_{\epsilon}}^{\lepsilon}\sum_{\hepsilon = 1_{\epsilon}}^{\lepsilon}\big(A_{(\alpha)}^{\iepsilon \kepsilon}A_{(\beta)}^{\hepsilon \pepsilon}+A_{(\alpha)}^{\iepsilon \pepsilon}A_{(\beta)}^{\hepsilon \kepsilon}-A_{(\alpha)}^{\hepsilon \pepsilon}A_{(\beta)}^{\iepsilon \kepsilon}-A_{(\alpha)}^{\hepsilon \kepsilon}A_{(\beta)}^{\iepsilon \pepsilon}\big)\partial_{\iepsilon}\Gamma_{\hepsilon}\\
{} +\!\sum_{\iepsilon = 1_{\epsilon}}^{\lepsilon}\sum_{\hepsilon = 1_{\epsilon}}^{\lepsilon}\!\big(A_{(\alpha)}^{\iepsilon \kepsilon}\partial_{\iepsilon}(A_{(\beta)}^{\hepsilon \pepsilon})+A_{(\alpha)}^{\iepsilon \pepsilon}\partial_{\iepsilon}(A_{(\beta)}^{\hepsilon \kepsilon})-A_{(\beta)}^{\iepsilon \kepsilon}\partial_{\iepsilon}(A_{(\alpha)}^{\hepsilon \pepsilon})-A_{(\beta)}^{\iepsilon \pepsilon}\partial_{\iepsilon}(A_{(\alpha)}^{\hepsilon \kepsilon})\big)\Gamma_{\hepsilon}=0.
\end{gather*}
It is now easily verified that the second double sum in the preceding expression vanishes identically as a consequence of~(\ref{cancel}), and that the first double sum is identically zero upon substitution of the expressions~(\ref{coeffmod}) of the coefficients $A_{(\alpha)}^{\igamma \jgamma}$.

Next we turn our attention to the condition (\ref{2ndordermore}), which we may rewrite, using~(\ref{BintermsofA}), as
\[
\sum_{\igamma=1_{\gamma}}^{\lgamma}\sum_{\hepsilon = 1_{\epsilon}}^{\lepsilon}\big(A_{(\alpha)}^{\igamma \jgamma}A_{(\beta)}^{\hepsilon \pepsilon}-A_{(\alpha)}^{\hepsilon \pepsilon}A_{(\beta)}^{\igamma \jgamma}\big)\big(\partial_{\igamma}\Gamma_{\hepsilon}-\partial_{\hepsilon}\Gamma_{\igamma}\big)=0.
\]
We have
\begin{gather}\label{curlGamma}
\partial_{\igamma}\Gamma_{\hepsilon}-\partial_{\hepsilon}\Gamma_{\igamma}=-\partial_{\igamma}\partial_{\hepsilon}\left(\log \frac{s^{\epsilon 1}}{\det S}\right)+\partial_{\hepsilon}\partial_{\igamma}\left(\log \frac{s^{\gamma 1}}{\det S}\right)=0,
\end{gather}
since the factors of $\det S$ cancel out in the logarithmic derivatives, and since $s^{\epsilon 1}$ (resp.~$s^{\gamma 1}$) is independent of the groups of variables ${\bf x}^{\epsilon}$ (resp.~${\bf x}^{\gamma}$).
Finally, we must show that the coefficients of the first derivatives in~(\ref{rawcomm}) are identically zero. We shall see that the analysis of these coefficients is slightly more involved than that of the second and third derivatives, and that the argument needed to prove their vanishing makes use of the generalized Robertson conditions.

Using (\ref{BintermsofA}) the vanishing of the coefficients of the first derivatives in~(\ref{rawcomm}) is seen to be equivalent to
\begin{gather}
\sum_{\igamma = 1_{\gamma}}^{\lgamma}\sum_{\jgamma = 1_{\gamma}}^{\lgamma}\sum_{\kepsilon =1_{\epsilon}}^{\lepsilon}\big(A_{(\alpha)}^{\igamma \jgamma}A_{(\beta)}^{\kepsilon \pepsilon}-A_{(\alpha)}^{\kepsilon \pepsilon}A_{(\beta)}^{\igamma \jgamma}\big)\partial_{\igamma}\partial_{\jgamma}\Gamma_{\kepsilon}\nonumber\\
\qquad{} +\sum_{\igamma = 1_{\gamma}}^{\lgamma}\sum_{\jgamma = 1_{\gamma}}^{\lgamma}\sum_{\kepsilon =1_{\epsilon}}^{\lepsilon}\big(\big[A_{(\alpha)}^{\igamma \jgamma}\partial_{\igamma}\partial_{\jgamma}A_{(\beta)}^{\kepsilon \pepsilon}-A_{(\beta)}^{\igamma \jgamma}\partial_{\igamma}\partial_{\jgamma}A_{(\alpha)}^{\kepsilon \pepsilon}\big]\Gamma_{\kepsilon}\nonumber\\
\qquad{} +2\big[A_{(\alpha)}^{\igamma \jgamma}\partial_{\igamma}A_{(\beta)}^{\kepsilon \pepsilon}-A_{(\beta)}^{\igamma \jgamma}\partial_{\igamma}A_{(\alpha)}^{\kepsilon \pepsilon}\big]\big)\partial_{\jgamma}\Gamma_{\kepsilon}\nonumber\\
\qquad{} -\sum_{\igamma = 1_{\gamma}}^{\lgamma}\sum_{\jgamma = 1_{\gamma}}^{\lgamma}\sum_{\kepsilon =1_{\epsilon}}^{\lepsilon}\sum_{\mgamma =1_{\gamma}}^{\lgamma}\big(A_{(\alpha)}^{\mgamma \jgamma}\partial_{\jgamma}\big(A_{(\beta)}^{\kepsilon \pepsilon}\Gamma_{\kepsilon}\big)-A_{(\beta)}^{\mgamma \jgamma}\partial_{\jgamma}\big(A_{(\alpha)}^{\kepsilon \pepsilon}\Gamma_{\kepsilon}\big)\big)=0.\label{1storderterm}
\end{gather}
The second triple sum in (\ref{1storderterm}) vanishes identically on account of (\ref{Thierry(a)}) and (\ref{cancel}). Likewise, using again (\ref{Thierry(a)}) and the expressions (\ref{coeffmod}) for $A_{(\alpha)}^{\igamma \jgamma}$, the condition (\ref{1storderterm}) reduces to
\begin{gather}\label{1storderfinal}
\sum_{\igamma = 1_{\gamma}}^{\lgamma}\sum_{\jgamma = 1_{\gamma}}^{\lgamma}\sum_{\kepsilon =1_{\epsilon}}^{\lepsilon}(\rho_{\alpha \gamma}\rho_{\beta \epsilon}-\rho_{\alpha \epsilon}\rho_{\beta \gamma}) \big(g^{\gamma}\big)^{\igamma \jgamma}(g^{\epsilon})^{\kepsilon \pepsilon}\big(\partial_{\igamma}\partial_{\jgamma}\Gamma_{\kepsilon}-\Gamma_{\igamma}(\partial_{\jgamma}\Gamma_{\kepsilon})\big)=0.
\end{gather}
In analogy with the notation used in \cite{BCR2-2002}, we introduce the tensor $\big({C^{i}}_{j}\big)$ defined by
\begin{gather}\label{defC}
{C^{\igamma}}_{\kepsilon}=\sum_{\jgamma = 1_{\gamma}}^{\lgamma}(\rho_{\alpha \gamma}\rho_{\beta \epsilon}-\rho_{\alpha \epsilon}\rho_{\beta \gamma})\big(g^{\gamma}\big)^{\igamma \jgamma}\partial_{\jgamma}
\Gamma_{\kepsilon}.
\end{gather}
The final steps in the proof will be to show that (\ref{1storderfinal}) is satisfied if an only if the tensor $\big({C^{i}}_{j}\big)$ has zero divergence, that is
\begin{gather*}
\sum_{\gamma =1}^{r}\sum_{\igamma = 1_{\gamma}}^{\lgamma}\nabla_{\igamma}{C^{\igamma}}_{\kepsilon}=0,
\end{gather*}
and that the generalized Robertson conditions are equivalent to ${C^{\igamma}}_{\kepsilon}=0$. These are the analogues for Painlev\'e metrics of the steps followed in~\cite{BCR2-2002} for the proof of the corresponding result for the special case of St\"ackel metrics satisfying the classical Robertson conditions.

We have
\begin{gather}
\sum_{\gamma =1}^{r}\sum_{\igamma = 1_{\gamma}}^{\lgamma}\nabla_{\igamma}{C^{\igamma}}_{\kepsilon}=\sum_{\gamma =1}^{r}\sum_{\igamma = 1_{\gamma}}^{\lgamma}\left(\partial_{\igamma}{C^{\igamma}}_{\kepsilon}+\sum_{\delta=1}^{r}\sum_{\hdelta =1_{\delta}}^{\ldelta}{\Gamma^{\igamma}}_{\jgamma \hdelta}{C^{\hdelta}}_{\kepsilon}\right.\nonumber\\
\left.\hphantom{\sum_{\gamma =1}^{r}\sum_{\igamma = 1_{\gamma}}^{\lgamma}\nabla_{\igamma}{C^{\igamma}}_{\kepsilon}=}{}
-\sum_{\delta=1}^{r}\sum_{\pdelta =1_{\delta}}^{\ldelta}{\Gamma^{\pdelta}}_{\jgamma \kepsilon}{C^{\hdelta}}_{\kepsilon}\right).\label{nablaC}
\end{gather}
Substituting into the expression (\ref{nablaC}) of the divergence the definition~(\ref{defC}) of the tensor~$\big({C^{i}}_{j}\big)$ and the expressions (\ref{Christoffel}) and (\ref{Christoffelmixed}) that were computed above for the Christoffel symbols of Painlev\'e metrics, we obtain
\begin{gather}
\sum_{\gamma =1}^{r}\sum_{\igamma = 1_{\gamma}}^{\lgamma}\nabla_{\igamma}{C^{\igamma}}_{\kepsilon} =\sum_{\gamma =1}^{r}\sum_{\igamma = 1_{\gamma}}^{\lgamma}\sum_{\pgamma =1_{\gamma}}^{\lgamma}(\rho_{\alpha \gamma}\rho_{\beta \epsilon}-\rho_{\alpha \epsilon}\rho_{\beta \gamma})\Bigg[(\partial_{\igamma}\big(g^{\gamma}\big)^{\igamma \pgamma})\partial_{\pgamma}\Gamma_{\kepsilon}\nonumber\\
\hphantom{\sum_{\gamma =1}^{r}\sum_{\igamma = 1_{\gamma}}^{\lgamma}\nabla_{\igamma}{C^{\igamma}}_{\kepsilon}=}{}
+\big(g^{\gamma}\big)^{\igamma \pgamma}\partial_{\igamma}\partial_{\pgamma}\Gamma_{\kepsilon} -\big(g^{\gamma}\big)^{\igamma \pgamma}(\partial_{\pgamma}\Gamma_{\kepsilon})\left(\partial_{\igamma}\log \frac{s^{\epsilon 1}}{\det S}\right)\nonumber\\
\hphantom{\sum_{\gamma =1}^{r}\sum_{\igamma = 1_{\gamma}}^{\lgamma}\nabla_{\igamma}{C^{\igamma}}_{\kepsilon}=}{}
-(\partial_{\pgamma}\Gamma_{\kepsilon})\left(\sum_{\hgamma =1_{\gamma}}^{\lgamma}\big(g^{\gamma}\big)^{\hgamma \pgamma}\Gamma_{\hgamma}+\partial_{\jgamma}\big(g^{\gamma}\big)^{\jgamma \pgamma}\right)\nonumber\\
\hphantom{\sum_{\gamma =1}^{r}\sum_{\igamma = 1_{\gamma}}^{\lgamma}\nabla_{\igamma}{C^{\igamma}}_{\kepsilon}=}{}
+\frac{1}{2}\sum_{\hepsilon =1_{\epsilon}}^{\lepsilon}\sum_{\kepsilon=1_{\epsilon}}^{\lepsilon}(g^{\epsilon})^{\pepsilon \hepsilon}(\partial_{\igamma} g_{\kepsilon \hepsilon})\big(g^{\gamma}\big)^{\igamma \jgamma}\big( \partial_{\jgamma}\Gamma_{\pepsilon}+\partial_{\pepsilon}\Gamma_{\jgamma}\big)\Bigg].\label{nablaCint}
\end{gather}
Substituting into (\ref{nablaCint}) the expression (\ref{Painleveblock}) of the block components of the metric, we get
\begin{gather}
\sum_{\gamma =1}^{r}\sum_{\igamma = 1_{\gamma}}^{\lgamma}\nabla_{\igamma}{C^{\igamma}}_{\kepsilon} =\sum_{\gamma =1}^{r}\sum_{\igamma = 1_{\gamma}}^{\lgamma}\sum_{\pgamma =1_{\gamma}}^{\lgamma}(\rho_{\alpha \gamma}\rho_{\beta \epsilon}-\rho_{\alpha \epsilon}\rho_{\beta \gamma})\big(g^{\gamma}\big)^{\igamma \pgamma}\bigg[\big(\partial_{\igamma}\partial_{\pgamma}\Gamma_{\kepsilon}-\Gamma_{\igamma}\partial_{\pgamma}\Gamma_{\kepsilon}\big)\nonumber\\
\hphantom{\sum_{\gamma =1}^{r}\sum_{\igamma = 1_{\gamma}}^{\lgamma}\nabla_{\igamma}{C^{\igamma}}_{\kepsilon}}{} +\frac{1}{2}\left(\partial_{\igamma}\left(\log \frac{s^{\epsilon 1}}{\det S}\right)\right)\big(\partial_{\kepsilon} \Gamma_{\pgamma}- \partial_{\pgamma}\Gamma_{\kepsilon}\big)\bigg].\label{nablaCeff}
\end{gather}
We now remark that the first derivative terms $\partial_{\kepsilon} \Gamma_{\pgamma}- \partial_{\pgamma}\Gamma_{\kepsilon}$ in (\ref{nablaCeff}) vanish identically on account of the identities~(\ref{curlGamma}), so that we finally obtain
\begin{gather}
\sum_{\gamma =1}^{r}\sum_{\igamma = 1_{\gamma}}^{\lgamma}\nabla_{\igamma}{C^{\igamma}}_{\kepsilon}=\sum_{\gamma =1}^{r}\sum_{\igamma = 1_{\gamma}}^{\lgamma}\sum_{\pgamma =1_{\gamma}}^{\lgamma}(\rho_{\alpha \gamma}\rho_{\beta \epsilon}-\rho_{\alpha \epsilon}\rho_{\beta \gamma})\big(g^{\gamma}\big)^{\igamma \pgamma}\nonumber\\
\hphantom{\sum_{\gamma =1}^{r}\sum_{\igamma = 1_{\gamma}}^{\lgamma}\nabla_{\igamma}{C^{\igamma}}_{\kepsilon}=}{}\times \big[\big(\partial_{\igamma}\partial_{\pgamma}\Gamma_{\kepsilon}-\Gamma_{\igamma}\partial_{\pgamma}\Gamma_{\kepsilon}\big)\big].\label{nablaCfinal}
\end{gather}
It therefore follows from (\ref{nablaCfinal}) that (\ref{1storderfinal}) will hold if and only if the tensor ${C^{\igamma}}_{\kepsilon}$ is divergence-free. Recapitulating our steps, we have shown that the vanishing of the commutator $[\Delta_{K_{(\alpha)}}\!{,}\Delta_{K_{(\beta)}}]\!$ is equivalent to the vanishing of the Poisson bracket $ \{K_{(\alpha)}, K_{(\beta)}\}$ and that of the divergence of~${C^{\igamma}}_{\kepsilon}$, that is,
\[
\big[\Delta_{K_{(\alpha)}},\Delta_{K_{(\beta)}}\big]=0 \iff \big\{K_{(\alpha)}, K_{(\beta)}\big\}=0\qquad \text{and} \qquad \sum_{\gamma =1}^{r}\sum_{\igamma = 1_{\gamma}}^{\lgamma}\nabla_{\igamma}{C^{\igamma}}_{\kepsilon}=0.
\]
The proof of the commutation relations~(\ref{commrelKill}) is concluded by observing that the generalized Robertson conditions~(\ref{GenRobertson}) are equivalent to $\big({C^{i}}_{j}\big)=0$ thanks to Lemma~\ref{RC-NewFormulation}. Finally, it is easily shown that the operators $T_{\alpha}$, $1\leq \alpha \leq r$ defined by~(\ref{eigenvalueeq}) are identical to the operators~$\Delta_{K_{(\alpha)}}$ defined by (\ref{KLaplacian}) by observing that the Killing tensor $\big(K_{(\alpha)}^{ij}\big)$ is block-diagonal, with components given by
\[
K_{(\alpha)}^{i_{\beta}j_{\beta}}=\frac{s^{\beta \alpha}}{\det S}\big(G^{\beta}\big)^{i_{\beta}j_{\beta}},\qquad K_{(\alpha)}^{i_{\beta}j_{\gamma}}=0\qquad \text{for}\quad \beta \neq \gamma.
\]

\subsection{Proof of Theorem \ref{maintheorem3}}
We start from the expression (\ref{Yamabe3}) of the Laplace--Beltrami operator for the conformally rescaled Painlev\'e metrics. We now rescale $v$ defined by (\ref{rescaledu}) according to
\begin{gather}\label{rescalv}
v=R w.
\end{gather}
So that the Helmholtz equation (\ref{Yamabe3}) when expressed in terms of $w$ becomes
\begin{gather}
\sum_{\beta=1}^{r}\left(\frac{s^{\beta 1}}{\det S}\right) \left(-\Delta_{G_{\beta}}+\sum_{\jbeta =1_{\beta}}^{\lbeta}\gamma^{\jbeta}\partial_{\jbeta}-2\sum_{\ibeta = 1_{\beta}}^{\lbeta}\sum_{\jbeta = 1_{\beta}}^{\lbeta}\big(G^{\beta}\big)^{\ibeta \jbeta}(\partial_{\ibeta}\log R) \partial_{\jbeta}\right)w\nonumber\\
\qquad{}  +\left(q_{g,c}-\lambda c^{4}+\frac{\Delta_{g}R}{R}\right)w=0.\label{Helmw}
\end{gather}
The idea behind $R$-separability is that one choose the conformal factor $c$ appearing in the conformal rescaling~(\ref{confresc}) and the scaling factor~$R$ appearing in~\ref{rescalv}) in such a way that the Helmholtz equation, when expressed in terms of $w$, becomes separable in the groups of variables~${\bf x}^{\beta}$, $1\leq \beta \leq r$ under a certain condition on the conformal factor~$c$. This is achieved in two steps, the first one being to choose the scaling factor~$R$ so as to eliminate the first derivative terms~$\gamma^{\jbeta}\partial_{\jbeta}w$ in~(\ref{Helmw}). This is equivalent to $R$ solving the overdetermined system of PDEs given by
\begin{gather}\label{elim}
2\sum_{\ibeta = 1_{\beta}}^{\lbeta}\big(G^{\beta}\big)^{\ibeta \jbeta}(\partial_{\ibeta}\log R)=\gamma^{\jbeta},\qquad 1_{\beta}\leq \jbeta \leq \lbeta .
\end{gather}
Using the expression (\ref{gammaup}) of the coefficients $\gamma^{\jbeta}$, we see that the system~(\ref{elim}) admits a~so\-lu\-tion~$R$ given by
\begin{gather*}
R=\frac{\big(s^{11}\big)^\frac{l_1}{4} \cdots \big(s^{r1}\big)^{\frac{l_r}{4}}}{(\det S)^{\frac{n-2}{4}}},
\end{gather*}
which is precisely the scaling factor $R$ given by~(\ref{Rdef}), and which we shall work with from now on. We may now compute the expression of $R^{-1} \Delta_{g}R$ appearing in~(\ref{Helmw}), using the expression~(\ref{Laplacian}) of the Laplace--Beltrami operator for Painlev\'e metrics and the fact that $R$ solves (\ref{elim}). We obtain after some calculations
\begin{gather}
\frac{\Delta_{g}R}{R}=\sum_{\beta=1}^{r}\left(\frac{s^{\beta 1}}{\det S}\right)\left[\frac{\Delta_{G_{\beta}}R}{R}+\sum_{\jbeta =1_{\beta}}^{\lbeta}\gamma^{\jbeta}\frac{\partial_{\jbeta}R}{R}\right]\nonumber\\
\hphantom{\frac{\Delta_{g}R}{R}}{}=\sum_{\beta=1}^{r}\frac{s^{\beta 1}}{\det S}\left[\frac{1}{2}\sum_{\jbeta =1_{\beta}}^{\lbeta}\partial_{\jbeta}\gamma^{\jbeta}-\frac{1}{4}\sum_{\ibeta = 1_{\beta}}^{\lbeta}\sum_{\jbeta = 1_{\beta}}^{\lbeta}(G_{\beta})_{\ibeta \jbeta}\gamma^{\ibeta}\gamma^{\jbeta}\right.\nonumber\\
\left.\hphantom{\frac{\Delta_{g}R}{R}=}{}+\frac{1}{4}\sum_{\jbeta =1_{\beta}}^{\lbeta}\partial_{\jbeta}\big(\log|G_{\beta}|\gamma^{\jbeta}\big)\right].\label{loglap}
\end{gather}
When the expression (\ref{loglap}) of $R^{-1} \Delta_{g}R$ is substituted into the Helmholtz equation (\ref{Helmw}) satisfied by $w$, the equation becomes
\begin{gather}
\sum_{\beta=1}^{r}\left(\frac{s^{\beta 1}}{\det S}\right)\left[-\Delta_{G_{\beta}}+\frac{1}{2}\sum_{\jbeta = 1_{\beta}}\partial_{\jbeta}\gamma^{\jbeta} -\frac{1}{4}\sum_{\ibeta = 1_{\beta}}^{\lbeta}\sum_{\jbeta = 1_{\beta}}^{\lbeta}(G_{\beta})_{\ibeta \jbeta}\gamma^{\ibeta}\gamma^{\jbeta}\right.\nonumber\\
\left.\qquad{} +\frac{1}{4}\sum_{\jbeta =1_{\beta}}^{\lbeta}\partial_{\jbeta}\big(\log|G_{\beta}|\gamma^{\jbeta}\big)\right]w+\big(q_{g,c}-\lambda c^{4}\big)w=0.\label{Helmwmod}
\end{gather}
To second step towards $R$-separability is to choose the conformal factor $c$ in such a way that the equation (\ref{Helmwmod}) for $w$ becomes manifestly separable in the groups of variables ${\bf x}^{\alpha}$, $1\leq \alpha \leq r$, which is achieved by requiring that $c$ satisfy the scalar nonlinear PDE
\begin{gather}
q_{g,c}-\lambda c^{4} = -a_1 + \sum_{\beta=1}^{r}\left(\frac{s^{\beta 1}}{\det S}\right)\left[-\frac{1}{2}\sum_{\jbeta = 1_{\beta}}^{\lbeta}\partial_{\jbeta}\gamma^{\jbeta} +\frac{1}{4}\sum_{\ibeta = 1_{\beta}}^{\lbeta}\sum_{\jbeta = 1_{\beta}}^{\lbeta}(G_{\beta})_{\ibeta \jbeta}\gamma^{\ibeta}\gamma^{\jbeta}\right.\nonumber\\
\left. \hphantom{q_{g,c}-\lambda c^{4} =}{} -\frac{1}{4}\sum_{\jbeta =1_{\beta}}^{\lbeta}\partial_{\jbeta}\big(\log|G_{\beta}|\gamma^{\jbeta}\big)+\phi_{\beta}\right],\label{eqncor}
\end{gather}
where $a_1$ is a constant and the $\phi_{\beta}=\phi_{\beta}\big({\bf x}^{\beta}\big)$, $1\leq \beta \leq r$, are arbitrary smooth functions of the group of variables ${\bf x}^{\beta}$. Using the definition (\ref{Yamabe2}) of $q_{g,c}$, we see that the PDE~(\ref{eqncor}) may indeed be rewritten in the form~(\ref{eqnc}). Note that the equation (\ref{Helmwmod}) takes the form (\ref{sepeqw}). This concludes the proof of Theorem~\ref{maintheorem3}.

\section{Perspectives and open problems}\label{Perspectives}

While the main results of our paper provide a convenient starting point from which to initiate the study of the anisotropic Calder\'on problem in manifolds with boundary endowed with Painlev\'e metrics, there are a number of questions that are left open in the above analysis and that call for further investigation, not just from a separation of variables point of view, but also in a more general differential geometric context. In particular, we would like to mention the following:
\begin{itemize}\itemsep=0pt
\item It would be worthwhile to obtain more examples in closed form of Painlev\'e metrics which are not of St\"ackel type and for which the generalized Robertson conditions (\ref{GenRobertson}) are satisfied. Given that the notion of a Painlev\'e metric can readily be formulated in an arbitrary signature and in particular in Lorentzian signature, it would be of particular interest to construct examples that would be solutions of the Einstein vacuum equations in four or higher dimensions.
\item One should be able to obtain an intrinsic characterization of the separable conformal deformations and $R$-separability of the Painlev\'e metrics, considered in Theorem \ref{maintheorem3}. Conformal Killing tensors should be a key component of such a characterization \cite{CR2006, KaMi1984}.
\item While Painlev\'e metrics are a generalization of St\"ackel metrics, which admit orthogonal local coordinates by definition, the St\"ackel form admits an extension to non-orthogonal coordinates \cite{KaMi1981}, an important Lorentzian example of which is given by the Kerr metric in General Relativity. It would be of interest to similarly extend the notion of a Painlev\'e metric to a non-orthogonal setting, where the expression of the metric (\ref{Painleveform}) defining the Painlev\'e form would be generalized to allow for the presence of cross terms between pairs groups of coordinates ${\bf x}^{\alpha}$ and ${\bf x}^{\beta}$ for $\alpha\neq \beta$. Again, some partial results in this direction, which apply to the $4$-dimensional Lorentzian case, appear in \cite{HaMa1978}, and suggest that non-orthogonal separability in this generalized sense would imply the existence of commuting Killing vectors, as is the case with the St\"ackel form \cite{KaMi1981, Woo1975}.
\item The considerations of the present paper are all local, but there exist global classification results for manifolds admitting St\"ackel metrics (see \cite{Ki1997} and the references therein). For example, it is shown that in dimension two, a compact manifold which admits a sufficiently generic St\"ackel metric must be diffeomorphic to the $2$-sphere, the real projective plane, the $2$-torus or the Klein bottle. It would be of interest to obtain analogues of these results for Painlev\'e metrics which are not St\"ackel.
\item It would be of interest to characterize in analogy with the St\"ackel case the scalar or vector potentials which are compatible with the separation into groups of variables of the Helmholtz equation in the class of Painlev\'e metrics.
\end{itemize}

Some of the above questions appear to be challenging, but progress on them would help to improve our understanding of the geometries in which separation of variables can be achieved in a broader and less restrictive sense than complete separation into ordinary differential equations.

\subsection*{Acknowledgements}
We thank Willard Miller, Jr., Askold Perelomov and Petar Topalov for having kindly answered our queries on Painlev\'e metrics. We are also grateful to Claudia Chanu for her helpful replies to our questions pertaining to $R$-separability and the Robertson conditions for St\"ackel metrics. Finally we thank the referees for their detailed comments and constructive suggestions. Thierry Daud\'e's research is supported by the JCJC French National Research Projects Horizons, No.~ANR-16-CE40-0012-01, Niky Kamran's research is supported by NSERC grant RGPIN 105490-2018 and Francois Nicoleau's research is supported by the French National Research Project GDR Dynqua.

\pdfbookmark[1]{References}{ref}
\LastPageEnding

\end{document}